\documentclass[10pt,conference,letterpaper]{IEEEtran}
\usepackage{times,amsmath,epsfig}

\usepackage{xspace}

\usepackage{multirow}
\usepackage{amsfonts}
\usepackage{graphicx}
\usepackage{algorithm}
\usepackage{algorithmic}
\usepackage{subfigure}

\usepackage{color}
\usepackage{url}
\usepackage{mdwlist}  
\usepackage{rotating}

\newtheorem{lemma}{Lemma}

\newcommand{\hide}[1]{}

\newcommand{\bit}{\begin{itemize}}
\newcommand{\eit}{\end{itemize}}
\newcommand{\ben}{\begin{enumerate}}
\newcommand{\een}{\end{enumerate}}

\newcommand{\simgraph}{\textsc{NetSimile}\xspace}
\newcommand{\fsm}{\textsc{FSM}\xspace}
\newcommand{\eig}{\textsc{EIG}\xspace}

\let\oldmarginpar\marginpar
\renewcommand\marginpar[1]{\-\oldmarginpar[\raggedleft\footnotesize
  \textcolor{red}{#1}]%
{\raggedright\footnotesize \textcolor{red}{#1}}} 

\begin{document}

\title{NetSimile: A Scalable Approach to Size-Independent Network Similarity}

\author{
{Michele Berlingerio{\small $^\ast$} \qquad\  Danai Koutra{\small
    $^\dag$} \qquad\ 
  Tina Eliassi-Rad{\small $^\ddag$} \qquad\  Christos Falousos{\small $^\dag$} }%
\vspace{1.6mm}\\
\begin{tabular}{ccc}
\begin{minipage}{150pt}
\center
\fontsize{10}{10}\selectfont\itshape
$~^\ast$IBM Research\\
Dublin, Ireland\\
\fontsize{9}{9}\selectfont\ttfamily\upshape
mberling@ie.ibm.com%
\end{minipage}
&
\begin{minipage}{150pt}
\center
\fontsize{10}{10}\selectfont\rmfamily\itshape
$~^\dag$Carnegie Mellon University\\
Pittsburgh, PA, USA\\
\fontsize{9}{9}\selectfont\ttfamily\upshape
\{danai,christos\}@cs.cmu.edu
\end{minipage}
&
\begin{minipage}{150pt}
\center
\fontsize{10}{10}\selectfont\rmfamily\itshape
$~^\ddag$Rutgers University\\
Piscataway, NJ, USA\\
\fontsize{9}{9}\selectfont\ttfamily\upshape
eliassi@cs.rutgers.edu
\end{minipage}
\end{tabular}
}

\maketitle

\begin{abstract}
Given a set of \emph{k} networks, possibly with different sizes and no overlaps in nodes or edges, how can we quickly assess similarity between them, without solving the node-correspondence problem?  Analogously, how can we extract a small number of descriptive, numerical features from each graph that effectively serve as the graph's ``signature''?  Having such features will enable a wealth of graph mining tasks, including clustering, outlier detection, visualization, etc.

We propose \simgraph\ -- a novel, effective, and scalable method for solving the aforementioned problem.  \simgraph has the following desirable properties: (a) It gives similarity scores that are size-invariant. (b) It is scalable, being linear on the number of edges for ``signature'' vector extraction.  (c) It does {\em not} need to solve the node-correspondence problem. We present extensive experiments on numerous synthetic and real graphs from disparate domains, and show \simgraph's superiority over baseline competitors.  We also show how \simgraph enables several mining tasks such as clustering, visualization, discontinuity detection, network transfer learning, and re-identification across networks.
\end{abstract}

\section{Introduction}
\label{sec:intro}

We address the problem of \emph{network similarity}.  Specifically, given a set of networks of (possibly) different sizes, and without knowing the node-correspondences, how can we efficiently provide a meaningful measure of structural similarity (or distance)? For example, how structurally similar are the SDM and SIGKDD co-authorship graphs? How does their structural similarity compare with the similarity between the SDM and ICDM co-authorship graphs?  Such measures are extremely useful for numerous graph-mining tasks. One such task is clustering: given a set of graphs, find groups of similar ones; conversely, find anomalies or discontinuities -- i.e., graphs that stand out from the rest. Transfer learning is another application, if graphs $G_1$ and $G_2$ are similar, we can transfer conclusions from one to the other to perform across-network classification with better classification accuracy. Yet another application is re-identification across two graphs -- where if the two graphs are similar, we can use information in one to re-identify nodes in the other. 

We define the \textbf{network similarity / distance} problem as follows.  \textbf{Input:} A set of $k$ anonymized networks of potentially different sizes, which may have no overlapping nodes or edges. \textbf{Output:} The structural similarity (or distance) scores of any pair of the given networks (or better yet, a feature vector for each network).\footnote{Throughout the paper, we assume similarity and distance are interchangeable.}

\begin{figure*}[th!]
\begin{tabular}{ccc}
\hspace{-6mm}\includegraphics[width=0.35\linewidth]{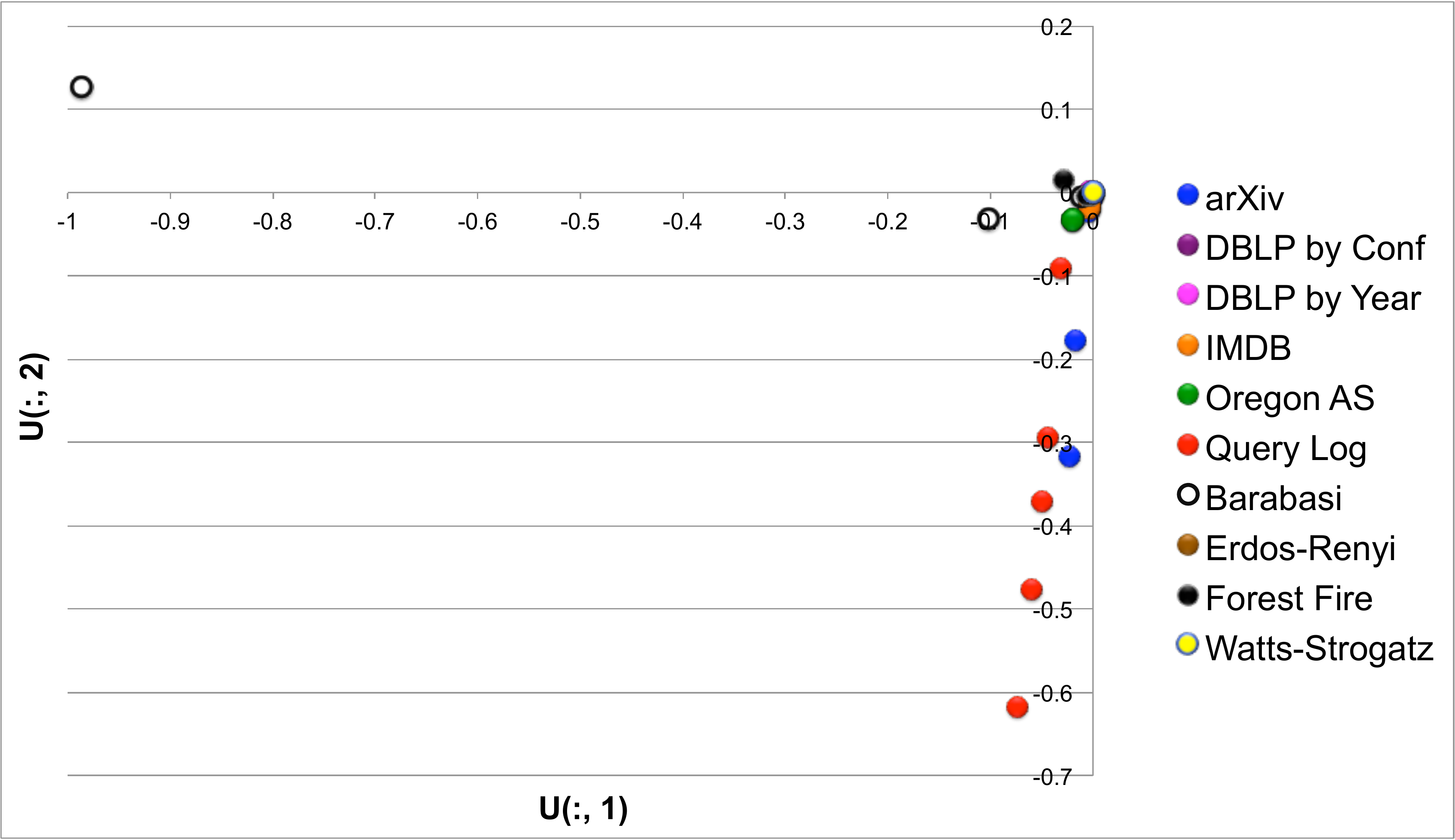} &
\hspace{-3.5mm}\includegraphics[width=0.35\linewidth]{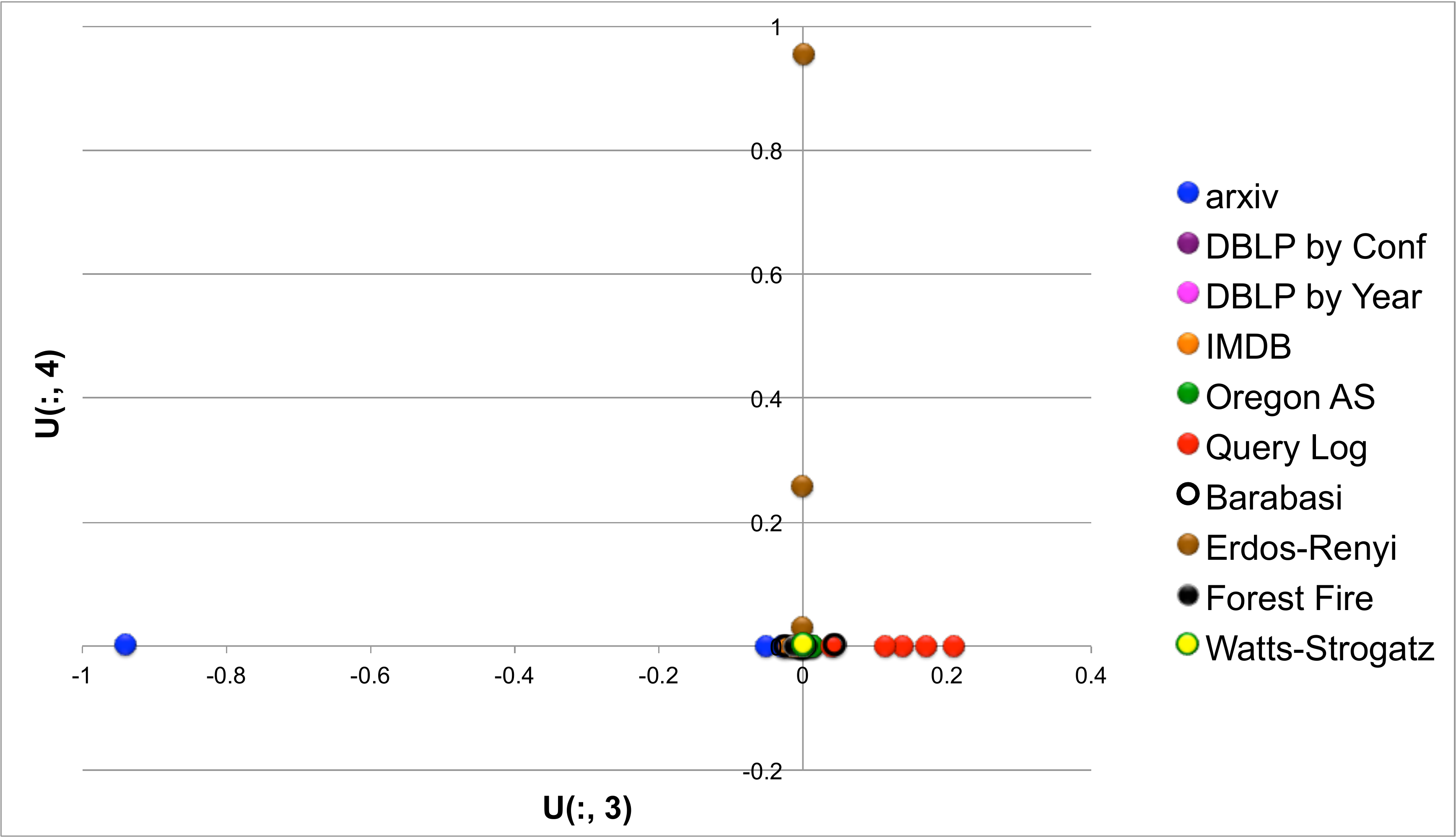} & 
\hspace{-3.5mm}\includegraphics[width=0.35\linewidth]{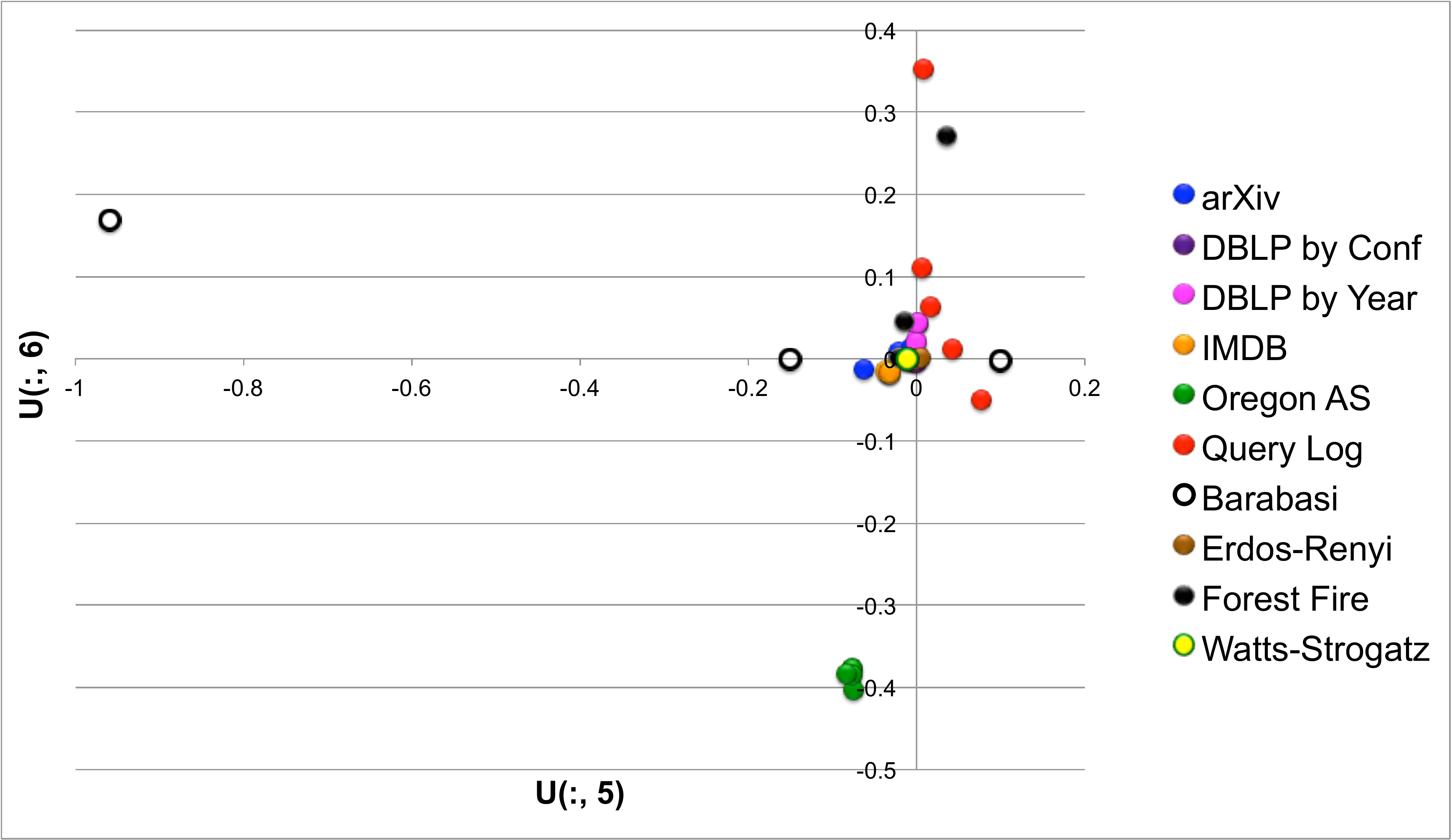}
\\
(a) 1st vs.~2nd principal components  & (b) 3rd vs.~4th principal components & (c) 5th vs.~6th principal components \\
of \simgraph features  & of \simgraph features  & of \simgraph features \\
\end{tabular}
\caption{An illustration of the visualization power of \simgraph's features. After doing SVD on the aggregated-feature matrix $W$ [$k$ graphs $\times$ $fr$ aggregated-features] (where $f$ = $\#$ features and $r$ = $\#$ aggregators), we get $k$ points in low-dimensional space, where similar graphs naturally cluster. (a) The query-log graphs (in red) fall on a line. (b) The Erd\"os-R\'enyi graphs (brown) fall on the vertical line. (c) The ``autonomous systems'' graphs (in green) group together.}\label{fig:svd}
\end{figure*}

The core of our approach, \simgraph, is a careful extraction and evaluation of structural features. For every graph $G$, we derive a small number of numerical features, which capture the topology of the graph as moments of distributions over its structural features.  The similarity score between two graphs then is just the similarity of their ``signature'' feature vectors. Once we have the similarity function, we can do a wealth of data mining tasks, including clustering, visualization, and anomaly detection.

Our empirical study includes experiments on more than 30 real-world networks and various synthetic networks generated by four different graph generators (namely, Erd\"os-R\'enyi, Forest Fire, Watts-Strogatz, and Barab\'asi Preferential Attachment).  We compare \simgraph with two baselines. The first baseline extracts frequent subgraphs from the given graphs and performs pairwise comparison on the intersection of the two sets of frequent patterns. The second baseline computes the \emph{k} largest eigenvalues of each network's adjacency matrix and measures the distance between them.  

Our experiments provide answers to the following questions:  How do the various methods compare w.r.t.~their similarity scores? Are their results intuitive (e.g., is a social network more similar to another social network than to a technological network)? How do they compare to null models?  Are the methods just measuring the sizes of the networks in their comparisons? How scalable are the various methods? Can we build a useful taxonomy for networks based on their similarities?

\textbf{Proof of Concept.} When measuring similarity, having a representative set of features for each graph is very powerful. The simplest way to illustrate this power is through visualization. Given a set of $k$ graphs (which we describe in Section~\ref{subsec:data}; and which include co-authorship networks, autonomous systems networks, Erd\"{o}s-R\'enyi graphs, etc), we can use \simgraph to extract a $graph \times feature$ matrix, and then project this matrix into its principal component space through Singular Value Decomposition (SVD).

Figure~\ref{fig:svd} depicts the scatterplots of the \simgraph $graph \times feature$ matrix' 1st vs.~2nd principal components, 3rd vs.~4th principal components, and 5th vs.~6th principal components.   As the plot clearly shows, \simgraph's graph ``signature'' vectors are discriminative enough to separate out the various types of graphs and to find anomalous graphs among the set of input graphs. For example, in Figure~\ref{fig:svd}(a), the synthetic ``Barab\'asi-100K-nodes'' stands out, while a group of ``query-log'' graphs align (red circles). In Figure~\ref{fig:svd}(b), the Erd\"{o}s-R\'enyi graphs (in brown) occupy the vertical axis, while in Figure~\ref{fig:svd}(c) the green points (``Oregon AS'' graphs) cluster together.

\paragraph{Contributions}

\bit
\item \emph{Novelty}:  By using moments of distribution as aggregators, \simgraph generates a single ``signature'' vector for each graph based on the local and neighborhood features of its nodes. 
  
\item \emph{Effectiveness}: \simgraph produces similarity / distance measures that are size-independent, intuitive, and interpretable. 
  
\item  \emph{Scalability}: The runtime complexity for generating \simgraph's ``signature'' vectors is linear on the number of edges. 

\item \emph{Applicability:} \simgraph's  ``signature'' vectors are useful in many graph mining tasks. 
\eit

The rest of the paper is organized into the following sections: Proposed Method, Experiments, Related Work, and Conclusions.

\section{Proposed Method}
\label{sec:meth}
Algorithm~\ref{wrapper} outlines \simgraph, which has three  steps: \emph{feature extraction} (Algorithm \ref{getfeatures}), \emph{feature aggregation} (Algorithm \ref{aggregator}), and \emph{comparison} (Algorithm \ref{compare}).

\textbf{Feature extraction.} \simgraph's feature extractor (see Algorithm~\ref{getfeatures}) generates a set of structural features for each node based on its local and egonet-based features -- a node's egonet is the induced subgraph of its neighboring nodes.   Specifically, \simgraph computes the following 7 features. 

\bit
\item $d_i=|N(i)|$: number of neighbors (i.e. degree) of node $i$; $N(i)$ denotes the neighbors of node $i$. 

\item $c_i$: clustering coefficient of node $i$, defined as the number of triangles connected to node $i$ over the number of connected triples centered on node $i$. 

\item $\bar{d}_{N(i)}$: average number of node $i$'s two-hop away neighbors, computed as $\frac{1}{d_i}\sum_{\forall j \in N(i)}^{}d_j$. 

\item $\bar{c}_{N(i)}$: average clustering coefficient of N($i$) \hide{node $i$'s neighbors}, calculated as $\frac{1}{d_i}\sum_{\forall j \in N(i)}^{}c_j$. 

\item $|E_{ego(i)}|$: number of edges in node $i$'s egonet; $ego(i)$ returns node $i$'s egonet.

\item $|E^o_{ego(i)}|$: number of outgoing edges from $ego(i)$. \hide{node $i$'s egonet.}

\item $|N(ego(i))|$: number of neighbors of $ego(i)$. \hide{node $i$'s egonet.}
\eit

Note that \simgraph is flexible enough to incorporate additional features. We choose these seven local and egonet-based features because they satisfy our constraints in terms of effectiveness (namely, size-independence, intuitiveness, and interpretability) and scalability (see Section~\ref{sec:exp}).  Also empirically, we observed the aforementioned features to be sufficient for measuring similarity across graphs from various domains (details in Section~\ref{sec:exp}).

\begin{algorithm}[t]
\caption{\simgraph}
\label{wrapper}
\small
\begin{algorithmic}[1]
\REQUIRE{\{$G_1, G_2, \cdots, G_k$\}, $doClustering$}
\STATE // extract features from nodes
\STATE \{$F_{G_1}, F_{G_2}, \cdots, F_{G_k}$\}:={\sc GetFeatures}(\{$G_1, G_2, \cdots, G_k$\})
\STATE // generate ``signature'' vectors for each graph
\STATE \{$\vec{s}_{G_1}, \vec{s}_{G_2}, \cdots, \vec{s}_{G_k}$\}:={\sc
  Aggregator}(\{$F_{G_1}, F_{G_2}, \cdots, F_{G_k}$\}) 
\STATE // do comparison and return similarity/distance values or clusterings for the given graphs
\RETURN {\sc Compare}(\{$\vec{s}_{G_1}, \vec{s}_{G_2}, \cdots, \vec{s}_{G_k}$\}, $doClustering$)
\end{algorithmic}
\end{algorithm}

\begin{algorithm}[t]
\caption{\simgraph's {\sc GetFeatures}}
\label{getfeatures}
\small
\begin{algorithmic}[1]
\REQUIRE{\{$G_1, G_2, \cdots, G_k$\}}
\FORALL{$j \in \{G_1, G_2, \cdots, G_k\}$}
\STATE $F_{G_j}$=$[]$  // initialize feature matrix for $G_j$
\STATE // extract features for all nodes in $G_j$
\FORALL{$i\in V_j$}  
\STATE $F_{G_j}$=$F_{G_j}$\ $\cup$ 
\STATE ~~~~~ $\{\{d_i, c_i, \bar{d}_{N(i)}, \bar{c}_{N(i)}, |E_{ego(i)}|, |E^o_{ego(i)}|, |N(ego(i))|\}\}$
\ENDFOR 
\ENDFOR 
\STATE // return a set of $node \times feature$ matrices
\RETURN $\{F_{G_1}, F_{G_2}, \cdots, F_{G_k}\}$
\end{algorithmic}
\end{algorithm}

\begin{algorithm}[t]
\caption{\simgraph's {\sc Aggregator}}
\label{aggregator}
\small
\begin{algorithmic}[1]
\REQUIRE{\{$F_{G_1}, F_{G_2}, \cdots, F_{G_k}$\}}
\FORALL{$j \in \{F_{G_1}, F_{G_2}, \cdots, F_{G_k}\}$}
\STATE $\vec{s}_{G_j}$=$[]$  // initialize ``signature'' vector for $G_j$
\STATE // for each feature column in $F_{G_j}$, compute a 
set of aggregates 
\FORALL{$feat \in F_{G_j}$} 
\STATE $\vec{s}_{G_j}$=$\vec{s}_{G_j}$ $\cup$ 
\STATE ~~~~~~~~ $\{median(feat), mean(feat), stdev(feat),$
\STATE ~~~~~~~~~ $skewness(feat), kurtosis(feat)\}$
\ENDFOR 
\ENDFOR
\STATE // return a set of ``signature'' vectors for the graphs
\RETURN \{$\vec{s}_{G_1}, \vec{s}_{G_2}, \cdots, \vec{s}_{G_k}$\}
\end{algorithmic}
\end{algorithm}

\begin{algorithm}[t]
\caption{\simgraph's {\sc Compare}}
\label{compare}
\small
\begin{algorithmic}[1]
\REQUIRE{\{$\vec{s}_{G_1}, \vec{s}_{G_2}, \cdots, \vec{s}_{G_k}$\}}, $doClustering$
\IF {$doClustering$}
\STATE // return clusterings of the given graphs
\RETURN {\sc Cluster}(\{$\vec{s}_{G_1}, \vec{s}_{G_2}, \cdots, \vec{s}_{G_k}$\})
\ELSE
\STATE // return similarity/distance values of the graphs
\RETURN {\sc PairwiseCompare}(\{$\vec{s}_{G_1}, \vec{s}_{G_2}, \cdots, \vec{s}_{G_k}$\})
\ENDIF
\end{algorithmic}
\end{algorithm}

\textbf{Feature aggregation.} After the feature extraction step, \simgraph has extracted a $node \times feature$ matrix, $F_{G_j}$, for each graph $G_j \in \{G_1, G_2, \cdots, G_k\}$.  We can measure similarity between graphs by comparing their feature matrices (see discussion below).  However, we discovered that generating a single ``signature'' vector for each graph produces more efficient and effective comparisons.  To this end, \simgraph uses the following five aggregators on each feature (i.e., on each column of $F_{G_j}$): \emph{median}, \emph{mean}, \emph{standard deviation}, \emph{skewness}, and \emph{kurtosis}.  Note that the latter four of the five aggregators are moments of distribution of each feature.  \simgraph is flexible enough to use other aggregators as well, though we found these five to be sufficient for the task of network comparison and satisfy our effectiveness and scalability constraints (see Section~\ref{sec:exp}).

\textbf{Comparison.} After the feature aggregation step, \simgraph\ has produced a ``signature'' vector $\vec{s}_{G_j}$ for every graph $G_j \in \{G_1, G_2, \cdots, G_k\}$.  \simgraph\ now has the whole arsenal of clustering techniques and pairwise similarity / distance functions at its disposal.  Amongst the collection of pairwise similarity / distance functions, we found Canberra Distance ($d_{Can}(P,Q)= \sum_{i=1}^d \frac{|P_i-Q_i|}{P_i+Q_i}$) to be very discriminative (a good property for a distance measure). This is because Canberra Distance is sensitive to small changes near zero; and it normalizes the absolute difference of the individual comparisons (see discussion in Section~\ref{sec:exp}). 

\textbf{Computational complexity.}  Let $k$ = number of graphs given to \simgraph (i.e., $k =|\{G_1,\cdots, G_k\}|$), $n_j$ = the number of nodes in $G_j$,  $m_j$ = the number of edges in $G_j$,  $f$ = number of structural features extracted, and  $r$ = number of aggregators used.\footnote{Note that $f$, $r$, and $k$ are in the 10s.}

\begin{lemma} The runtime complexity for generating \simgraph's ``signature'' vectors is linear on the number of edges in $\{G_1, \cdots, G_k\}$, and specifically
\begin{equation}	
    O(\sum_{j=1}^k (f n_{j} + f n_j log(n_j) ))
\end{equation}

where $f \ll n_j \ll m_j$ and $n_j log(n_j) \approx m_j$ in real-world graphs.\end{lemma}
\begin{proof}  To generate \simgraph's ``signature'' vectors, features need to be extracted and then aggregated.

\emph{Feature Extraction:}~Recall that \simgraph is computing local and neighborhood-based structural features.  As proved in~\cite{Refex11}, computation of neighborhood-based features is expected to take $O(n_j)$ for real-world graphs.  Therefore to compute $f$ neighborhood-based features on a graph $G_j$, it takes $O(f n_j)$.

\emph{Feature Aggregation:}~This is $O(f n_j log(n_j))$ for each graph $G_j$.  Recall that \simgraph's aggregators are median, mean, standard deviation, skewness, and kurtosis.  The latter four can be computed in one-pass through the $f$ feature values.  The most expensive computation is the median which cannot be done in one-pass.  However, it can be computed in $O(n log(n) + n)$ for $n$ numbers.  Basically, one needs $O(n log(n))$ to sort the $n$ numbers.  Then, a selection algorithm can be used to get the median with only O(n) operations.
\end{proof}

\textbf{Remark: Network comparison through statistical hypothesis testing.} Given the $node \times feature$ matrices of two graphs, $F_{G_1}$ and $F_{G_2}$, \simgraph can use statistical hypothesis testing to see if the two graphs are samples from the same underlying distribution.  Specifically, \simgraph normalizes each column (i.e. feature) in $F_{G_1}$ and $F_{G_2}$ by its $L_2$ norm.  Then, \simgraph does pairwise hypothesis testing across the features of the graphs.  For example, it does hypothesis testing between the degree columns in $G_1$ and $G_2$; between the clustering coefficient columns in $G_1$ and $G_2$; and so on.  This process produces seven $p$-$values$ (corresponding to the seven features extracted by \simgraph).  To decide whether the two graphs are from the same underlying distribution, \simgraph uses the maximum $p$-$value$.  We also tried the average of the $p$-$values$, though that analysis did not produce as discriminative results as the maximum $p$-$value$. 

For the statistical hypothesis tests, \simgraph can use any test available. We tried the Mann-Whitney Test~\cite{mannwhitney} and the Kolmogorov-Smirnov Test~\cite{kstest}.  The Mann-Whitney Test is nonparametric.  It assumes two samples are independent and measures whether the two samples of observations have equally large values. The Kolmogorov-Smirnov Test is also nonparametric.  We used the two-sample Kolmogorov-Smirnov Test which compares two samples w.r.t.~the location and shape of the empirical cumulative distribution functions of the two samples.  We found that neither test generated enough discriminative power\footnote{We informally define discriminative power to be the power to make fine distinctions. More details in Section~\ref{sec:exp}.} to effectively capture differences between graphs (though the Mann-Whitney Test was more discriminative).  See Section~\ref{sec:exp} for details. 
 
\textbf{Remark: Network comparison at the local- vs. global-level.} Whether one prefers local-level network similarity to global-level network similarity depends on the application for which the similarity is being used.  \simgraph is designed such that it can take either local-level or global-level features.  Here, we emphasis \simgraph's local-level network similarity.   The advantages of local-level comparison is that node-level and egonet-level features are often more interpretable than global features -- e.g., consider average degree of a node vs. the number of distinct eigenvalues of the adjacency matrix.  Also, local-level features are computationally less expensive than global-level features -- e.g., consider clustering coefficient of a node vs. diameter of the graph. Moreover, looking at local-level features answers the question: ``are the given two networks from similar linking models?''  For example, consider the Facebook and Google+ social networks.  Even though Google+ is a smaller network than Facebook, are its users linking in a similar way to the users of the Facebook network?  In other words, is the smaller Google+ network following a similar underlying model as the lager Facebook network?  Local-level features can capture any similarity present in the linking models of the two networks, but global-level features cannot.

\section{Experiments}
\label{sec:exp}

This section is organized as follows.  First, we outline the real and synthetic datasets used in our experiments, as well as our experimental setup.  Second, we describe two baseline methods ``FSM'' (Frequent Subgraph Mining) and ``EIG'' (Eigenvalues Extraction).  Third, we present results that answer the following questions: How do the different approaches compare? Is there a particular method which clearly outperforms the others? If yes, to which extent? How can we interpret the results? Can we build a taxonomy over the networks based on our results? Is \simgraph affected by the sizes of the networks? How do the proposed methods scale? How well does \simgraph perform in various graph mining applications?

\subsection{Data and Experimental Setup}
\label{subsec:data}

\begin{table}[t!]
\centering\footnotesize\setlength{\tabcolsep}{1.5pt}
\scalebox{0.91}{
\begin{small}
\begin{tabular}{|c|c||r|r|r|r|r|r|r|}
\hline
\multicolumn{2}{|c||}{\bf Network} &  \multicolumn{1}{c|}{\bf $|\text{V}|$} &  \multicolumn{1}{c|}{\bf $|\text{E}|$} &  \multicolumn{1}{c|}{\bf $k$} &  \multicolumn{1}{c|}{\bf Net c} &  \multicolumn{1}{c|}{\bf $\bar{\text{c}}$} &  \multicolumn{1}{c|}{\bf $|\text{LCC}|$} & \multicolumn{1}{c|}{\bf \#CC}\\
\hline 
\hline
\multirow{5}{*}{\begin{sideways}arXiv\end{sideways}} 
& a-AstroPh & 18,772 & 396,160 & 42.21 & 0.318 & 0.677 & 17,903 & 290\\
& a-CondMat & 23,133 & 186,936 & 16.16 & 0.264 & 0.706 & 21,363 & 567\\
& a-GrQc & 5,242 & 28,980 & 11.06 & 0.630 & 0.687 & 4,158 & 355\\
& a-HepPh & 12,008 & 237,010 & 39.48 & 0.659 & 0.698 & 11,204 & 278\\
& a-HepTh & 9,877 & 51,971 & 10.52 & 0.284 & 0.600 & 8,638 & 429\\\hline
\multirow{6}{*}{\begin{sideways}DBLP-C\end{sideways}}
& d-vldb & 1,306 & 3,224 & 4.94 & 0.597 & 0.870 & 769 & 112\\
& d-sigmod & 1,545 & 4,191 & 5.43 & 0.601 & 0.856 & 1,092 & 116\\
& d-cikm & 2,367 & 4,388 & 3.71 & 0.560 & 0.873 & 890 & 361\\
& d-sigkdd & 1,529 & 3,158 & 4.13 & 0.505 & 0.879 & 743 & 189\\
& d-icdm & 1,651 & 2,883 & 3.49 & 0.518 & 0.887 & 458 & 281\\
& d-sdm & 915 & 1,501 & 3.28 & 0.540 & 0.870 & 243 & 165\\\hline
\multirow{5}{*}{\begin{sideways}DBLP-Y\end{sideways}}
& d-05 & 39,357 & 79,114 & 4.02 & 0.415 & 0.642 & 29,458 & 3,229\\
& d-06 & 44,982 & 94,274 & 4.19 & 0.379 & 0.632 & 35,223 & 3,140\\
& d-07 & 47,465 & 103,957 & 4.38 & 0.373 & 0.628 & 38,048 & 3,078\\
& d-08 & 47,350 & 107,643 & 4.55 & 0.378 & 0.612 & 38,979 & 2,849\\
& d-09 & 45,173 & 102,072 & 4.52 & 0.331 & 0.595 & 36,767 & 2,920\\\hline
\multirow{5}{*}{\begin{sideways}IMDb\end{sideways}}
& i-05 & 13,805 & 130,295 & 18.88 & 0.506 & 0.774 & 13,075 & 258\\
& i-06 & 14,228 & 142,955 & 20.09 & 0.480 & 0.760 & 13,458 & 269\\
& i-07 & 13,989 & 133,930 & 19.15 & 0.476 & 0.757 & 13,091 & 256\\
& i-08 & 14,055 & 132,007 & 18.78 & 0.469 & 0.750 & 13,313 & 273\\
& i-09 & 14,372 & 128,926 & 17.94 & 0.442 & 0.728 & 13,601 & 277\\\hline
\multirow{5}{*}{\begin{sideways}Query Log\end{sideways}}
& ql-1 & 138,976 & 1,102,606 & 15.87 & 0.055 & 0.599 & 132,012 & 3,238\\
& ql-2 & 108,420 & 876,517 & 16.17 & 0.055 & 0.594 & 103,095 & 2,482\\
& ql-3 & 89,406 & 707,579 & 15.83 & 0.053 & 0.588 & 85,246 & 1,941\\
& ql-4 & 75,838 & 582,703 & 15.37 & 0.051 & 0.583 & 72,396 & 1,600\\
& ql-5 & 42,946 & 253,469 & 11.80 & 0.047 & 0.573 & 40,691 & 1,027\\ \hline
\multirow{5}{*}{\begin{sideways}Oregon AS\end{sideways}}
& o-1 & 10,900 & 31,181 & 5.72 & 0.039 & 0.501 & 10,900 & 1\\
& o-2 & 11,019 & 31,762 & 5.76 & 0.040 & 0.495 & 11,019 & 1\\
& o-3 & 11,113 & 31,435 & 5.66 & 0.034 & 0.490 & 11,113 & 1\\
& o-4 & 11,260 & 31,304 & 5.56 & 0.032 & 0.487 & 11,260 & 1\\
& o-5 & 11,461 & 32,731 & 5.71 & 0.037 & 0.494 & 11,461 & 1\\
\hline
\end{tabular}
\end{small}
}
\caption{Real networks: \#nodes, \#edges, avg degree, network clustering coefficient (transitivity),  avg node clustering coefficient, \#nodes in the largest connected component, \#connected components.}
\label{tab:networks}
\end{table}

\textbf{Real Networks.} Table \ref{tab:networks} lists the basic statistics of the real networks used in our experiments. Here is a short description of each network.

\bit
\item \textbf{arXiv} (\url{http://arxiv.org}): Five different co-authorship networks corresponding to the following fields: Astro Physics, Condensed Matter, General Relativity, High Energy Physics and High Energy Physics Theory.

\item \textbf{DBLP-C} (\url{http://dblp.uni-trier.de}): Six different co-authorship networks from VLDB, SIGKDD, CIKM, ICDM, SIGMOD and WWW conferences, each spanning over 5 years (2005-2009).

\item \textbf{DBLP-Y}: Five different co-authorship networks, each corresponding to one of the years from 2005 to 2009 and consisting of data from 31 conferences.

\item \textbf{IMDb} (\url{http://www.imdb.com}): Five collaboration networks for movies issued from 2005 to 2009.  Each node represents a person who took part in the movie (i.e., cast and crew). Edges connect people who collaborated on a movie. 
 
\item \textbf{QueryLog} (\url{http://www.gregsadetsky.com/aol-data}):  Five word co-occurrence networks built from a query-log of approximately 20 millions web-search queries submitted by 650,000 users over 3 months.
  
\item \textbf{Oregon AS} (\url{http://snap.stanford.edu/data/}): 
 Five Autonomous Systems (AS) routing graphs between March 31st and May 26th 2001. 
\eit

\textbf{Synthetic Networks.} Apart from the real networks, we also produced several synthetic networks by using the following generators from the igraph library (\url{http://igraph.sourceforge.net}): 

\bit
\item \textbf{Barab\'{a}si-Albert}~\cite{generator-BA}: With a non-assortative version of the generator, we created graphs with 1K, 10K, and 100K nodes, adding 4 edges in each iteration.

\item \textbf{Forest-Fire}~\cite{generator-FF}: We generated graphs of size 1K, 10K, and 100K nodes, with 20\% forward burning probability, 40\% backward burning probability, and 4 ambassador vertices.
 
\item \textbf{Erd\"{o}s-R\'enyi}~\cite{generator-ER}: We used the $G(n, m)$ generator, where $n$ is the number of nodes and $m$ the number of edges,  and produced graphs $G(n, 2n)$ with 1K, 10K, and 100K nodes.

\item \textbf{Watts-Strogatz}~\cite{generator-WS}: We built graphs of size 200, 2K, and 20K nodes by setting the lattice dimension to 1, the degree to 4, and the rewiring probability to 0.3.
\eit

For each generator and for each node-set size, we built five networks. Our results report the average values obtained across the five networks per generator and node-set size.

\textbf{Experimental Setup.} We implemented our approach in C++ and Matlab, making use of the GNU Statistic Libraries and igraph. The code was run on a server equipped with 8 Intel Xeon processors at 3.0GHz, with 16GB of RAM, and running CentOS 5.2 Linux.

\subsection{Baseline Methods}
We compare \simgraph with (a) Frequent Subgraph Mining and (b) Eigenvalues Extraction. We chose these two methods because they are intuitive and widely applicable.  Many methods discussed in Section \ref{sec:background} are application-dependent.

\textbf{FSM (Frequent Subgraph Mining):} Given two graphs, we take the intersection of their frequent pattern-sets and build two vectors (one per graph) of relative supports of their patterns~\cite{Berlingerio09}.  We compare these \fsm vectors with \simgraph's ``signature'' vectors using Cosine Similarity and Canberra Distance. A clear drawback of \fsm is its lack of scalability (since it relates to subgraph isomorphism).

\textbf{EIG (Eigenvalues Extraction):} This is an intuitive measure of network similarity that is based on \emph{global} feature extraction (as opposed to the \emph{local} feature extraction of \simgraph). For each graph, we compute the $k$ largest eigenvalues\footnote{We tried a few values for $k$ and saw no significant changes around $10$; so we selected $k=10$.} of its adjacency matrix, and thus we obtain a vector of size $k$ per graph.  Then, we use the Canberra Distance in order to compare these vectors and find the pairwise similarities between the graphs.  A disadvantage of \eig is that it is size dependent: larger networks - or ones with larger LCC (Largest Connected Component) - have higher eigenvalues. Thus, \eig will lead to higher similarity between networks with comparable sizes. Moreover, there is no global upper-bound for eigenvalues, making distance values hard to compare.

\begin{table}[th!]
\centering\small
\begin{tabular}{|l||c|c|c|}\hline
\multicolumn{1}{|c||}{\bf Desired Properties} & {\bf \simgraph } & {\bf FSM} & {\bf EIG } \\\hline  \hline
{\bf Scalable} & \checkmark &  &  \checkmark  \\
{\bf Size-independent} &  \checkmark &  \checkmark &  \\
{\bf Intuitive} &  \checkmark & \checkmark &  \checkmark\\
{\bf Interpretable} & \checkmark & \checkmark & \\\hline
\end{tabular}\vspace{-2mm}
\caption{Properties of \simgraph and baselines}\label{tab:properties}
\vspace{-2mm}\end{table}

\begin{figure*}
\centering\footnotesize
\begin{tabular}{ccc}
\hspace{-8mm}\includegraphics[clip=true,trim=0 0 0 20, scale=0.5]{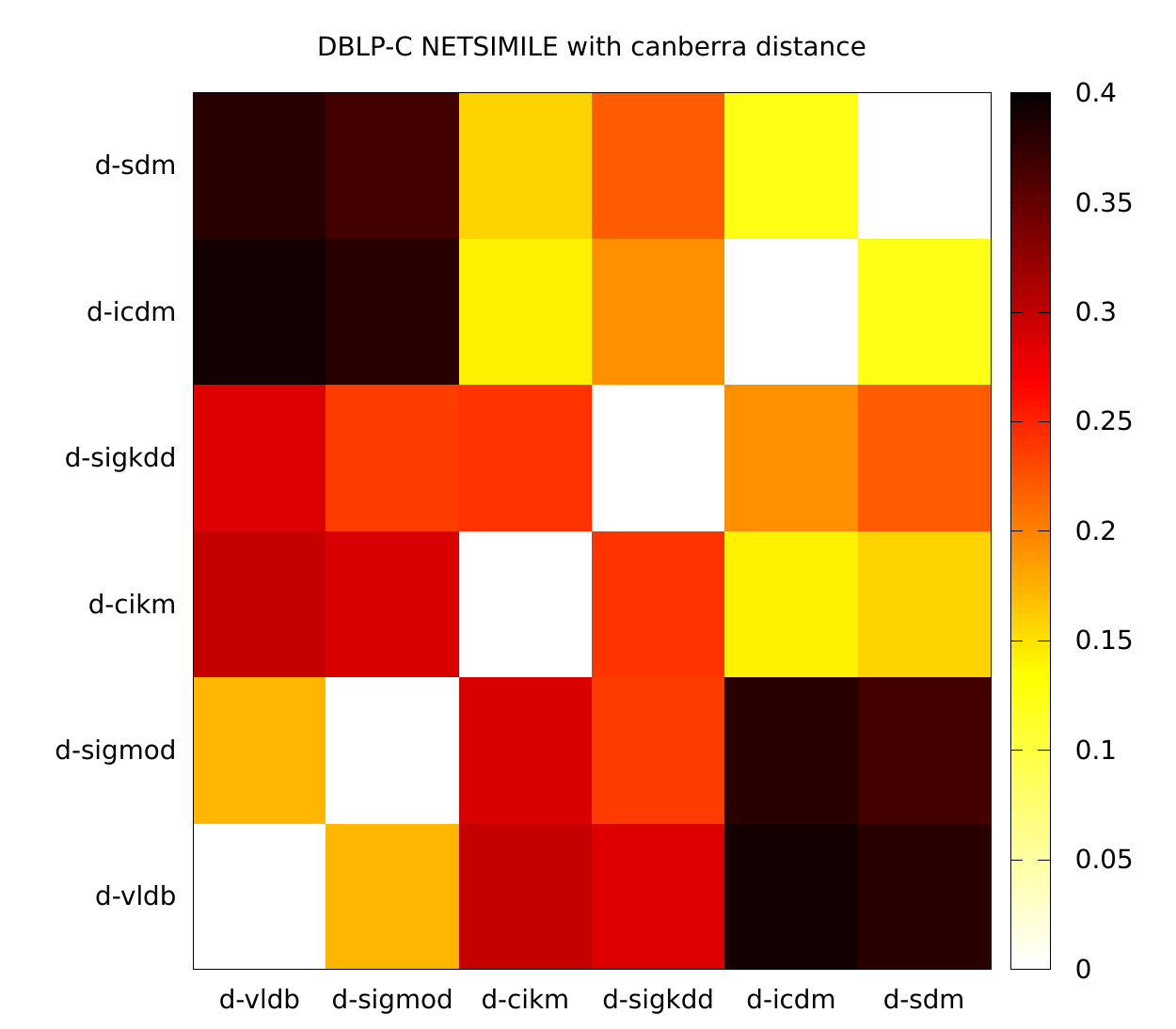} &
\hspace{-8mm}\includegraphics[clip=true,trim=15 0 0 20, scale=0.5]{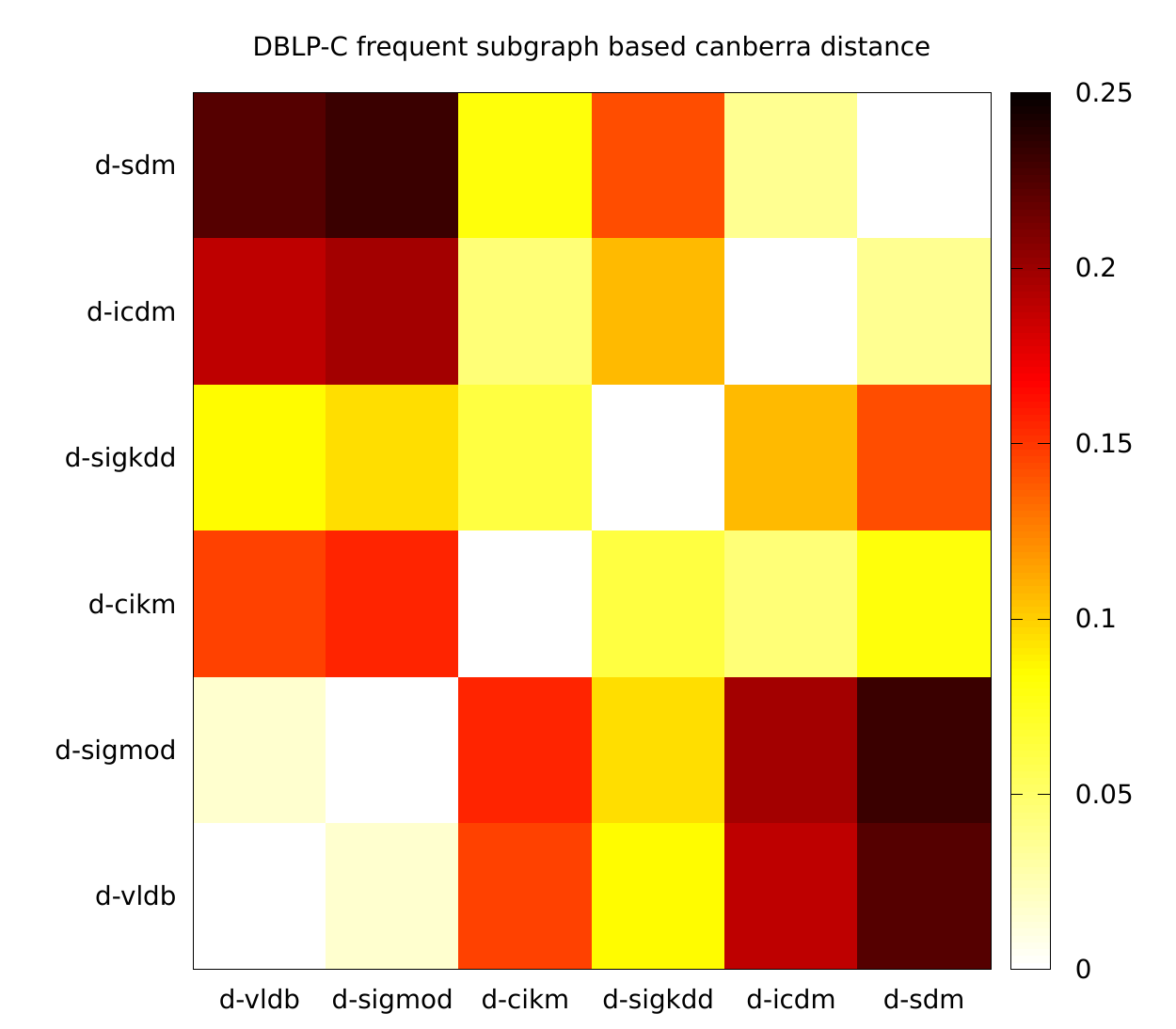} &
\hspace{-8mm}\includegraphics[clip=true,trim=15 0 0 20, scale=0.5]{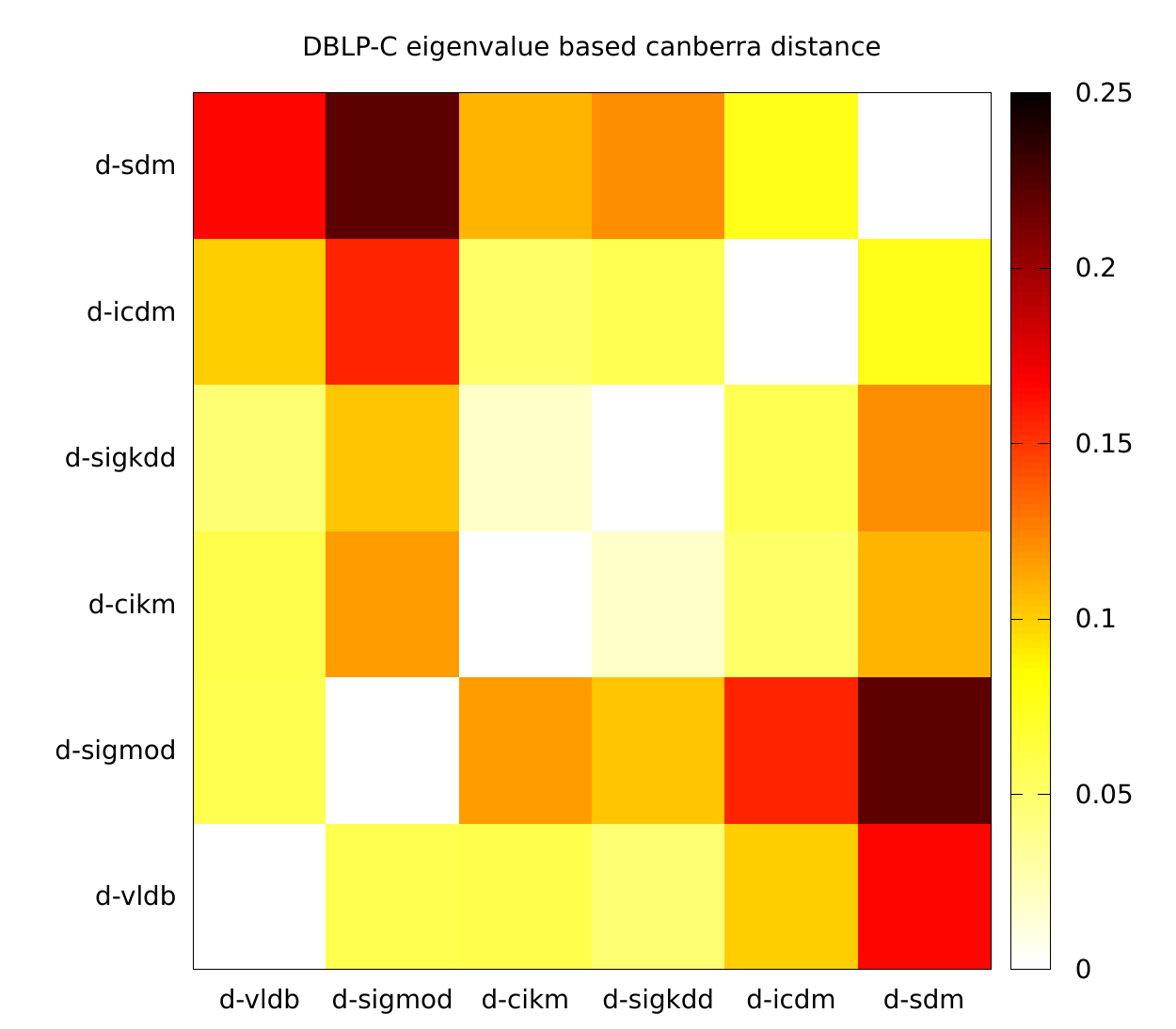} \\
\hspace{-8mm}(a) \simgraph Canberra Distance on DBLP-C    &
\hspace{-8mm}(b) \fsm Canberra Distance on DBLP-C  &
\hspace{-8mm}(c) \eig Canberra Distance on DBLP-C  \\
\hspace{-8mm}\includegraphics[clip=true,trim=0 0 0 20, scale=0.5]{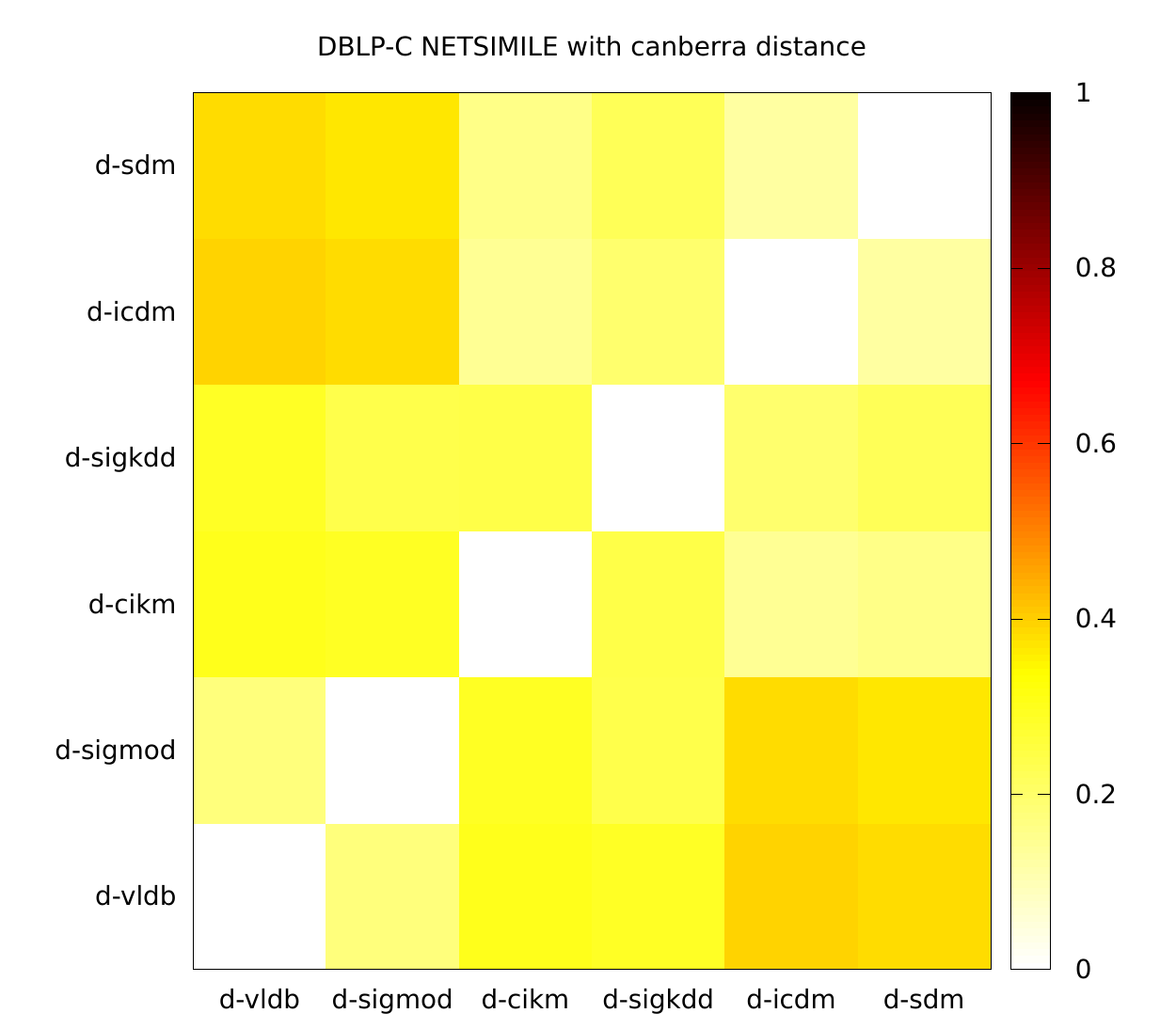} &
\hspace{-8mm}\includegraphics[clip=true,trim=0 0 0 20, scale=0.5]{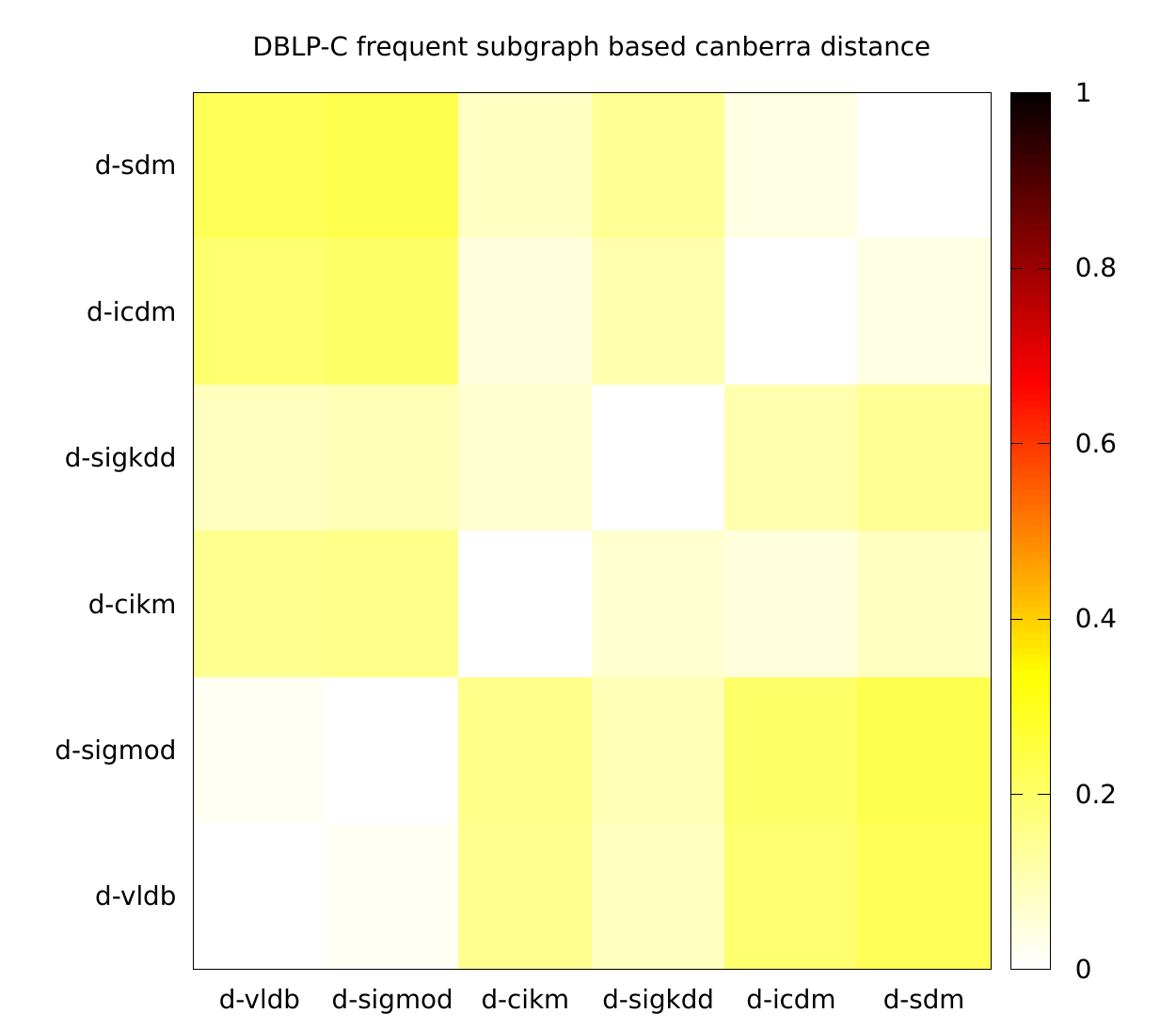} &
\hspace{-8mm}\includegraphics[clip=true,trim=0 0 0 20, scale=0.5]{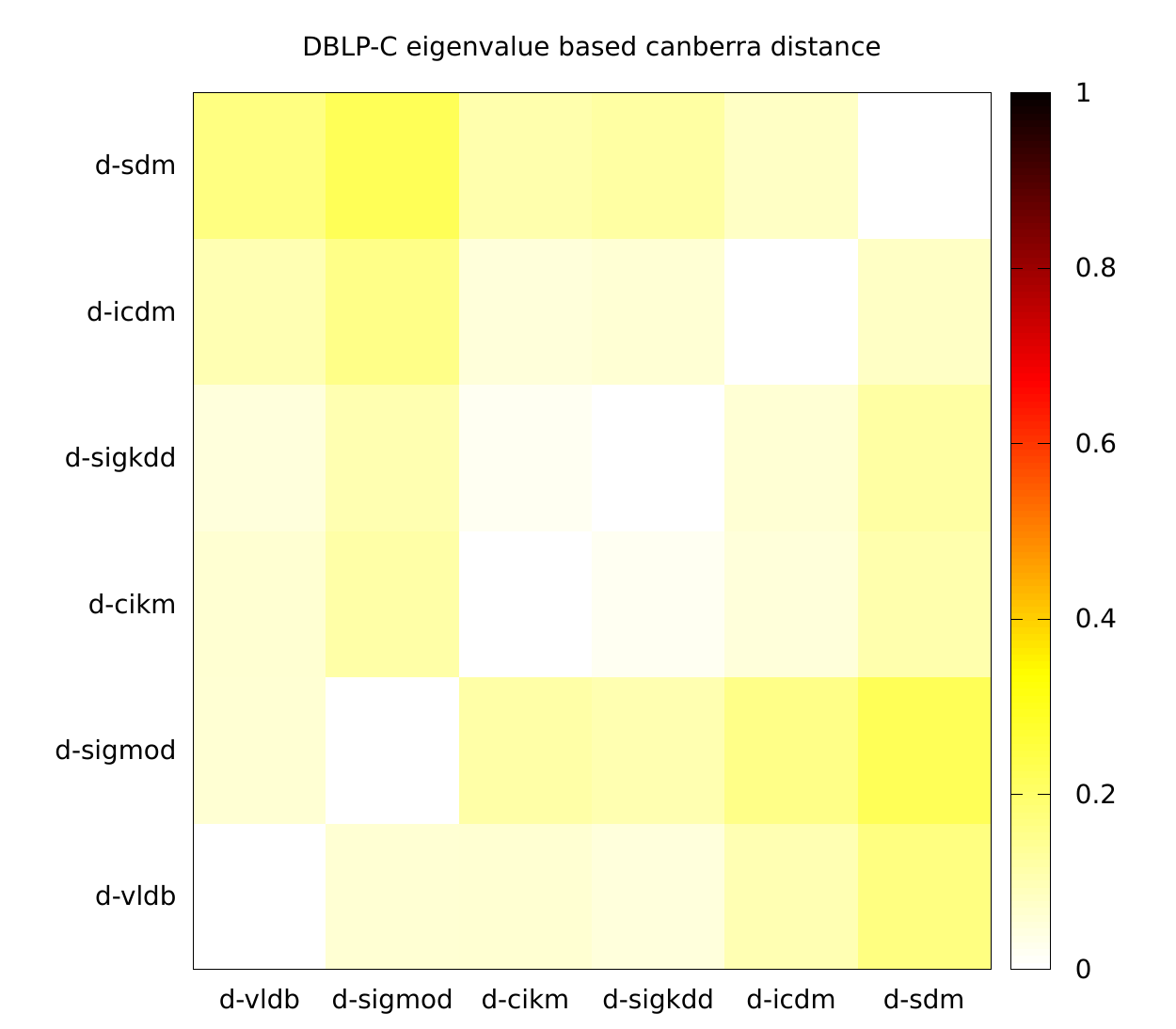} \\
\hspace{-2mm}(d) \simgraph Scaled Canberra Distance on DBLP-C   &
\hspace{-6mm}(e) \fsm Scaled Canberra Distance on DBLP-C    &
\hspace{-7mm}(f) \eig Scaled Canberra Distance on DBLP-C    \\
\end{tabular}
\caption{Comparative results: Canberra Distance scores between the DBLP-C networks by \simgraph (1st col.),  FSM (2nd col.), and  EIG (3rd col.). Second row is the scaled Canberra Distance $\in [0,1]$.
The baseline methods do not have the discriminative power of \simgraph.}\label{fig:coscandblpscaled}
\vspace{-3mm}\end{figure*}

\subsection{Comparative Results}
Table \ref{tab:properties} summarizes the basic properties of \simgraph, \fsm, and \eig.   \simgraph is the only one having the desired properties of being scalable, size-independent, intuitive, and interpretable. Here, being intuitive means providing results that are comparable with the background knowledge that we have about the data. For example, intuitively, a co-authorship network should be more similar to another co-authorship network than to a technological network like autonomous systems. Moreover, interpretability is given by the low complexity of the theory behind the concepts that build \simgraph. In fact, the signature vectors returned by \simgraph are comparing the moments of distributions of local (i.e. neighborhood-based) features of the two graphs, and these are commonly understandable measures. Empirical support follows next.

For each method (\simgraph, \fsm, and \eig), after extracting features from the graphs and obtaining one (aggregated) feature vector per graph, we apply the Canberra Distance.\footnote{For brevity, we do not report results on Cosine Similarity and the other similarity/distance measures, which we tried.  Most results are highly correlated \cite{survey_sim_measures}.} We also report results of the Mann-Withney U test on \simgraph's local feature distributions.  This statistical test is unsuitable for \eig and non-trivial for \fsm. 

Figure~\ref{fig:coscandblpscaled} depicts the results of a set of experiments involving the Canberra Distance on the DBLP-C datasets. The columns in Figure~\ref{fig:coscandblpscaled} correspond to the heatmaps we obtained from \simgraph, \fsm and \eig, respectively. The first row reports the results from the Canberra Distance. The second row reports the results from the {\em scaled} Canberra Distance, where each value is in [0,1].

Inspecting Figures~\ref{fig:coscandblpscaled}(a)-(b), we observe that the results of \simgraph are similar to \fsm.  For instance, according to both, the d-vldb network is similar to d-sigmod, which is not so similar to d-sdm.  However, the discriminative power of \simgraph becomes evident when we inspect Figures~\ref{fig:coscandblpscaled}(d)-(e).  In these figures, the Canberra Distance is scaled, so the results can be compared on an equal footing.  We observe that the scaled values of \fsm are much less discriminative.

Reexamining Figure \ref{fig:coscandblpscaled}, we also observe that the results from \eig differ from the ones from \simgraph and \fsm. According to \eig, d-vldb has no significant differences with d-sigmod, d-cikm, and d-sigkdd; while \simgraph and \fsm found differences. Moreover, there is no global normalization for the \eig values.\footnote{It is possible to do pairwise normalization by the number of nodes, but this is not general for any set of networks.}  Thus, global comparisons of a set of networks are harder to interpret with \eig than with \simgraph and \fsm.

As another point of comparison, we measure the entropy in feature vectors generated by \simgraph, \fsm, and \eig on the DBLP-C co-authorship networks.  As Figure~\ref{fig:entropy} shows, \simgraph's feature vectors have higher entropy than \fsm's or \eig's.  Higher entropy means more uncertainty (i.e., we need more bits to store the desired information).  So, \simgraph's feature vectors capture the nuances (i.e. uncertainty) in the graphs bettern than \fsm or \eig, which then leads to more discriminative power when comparing graphs.

\begin{figure}[!h]
\centering
\includegraphics[scale=0.7]{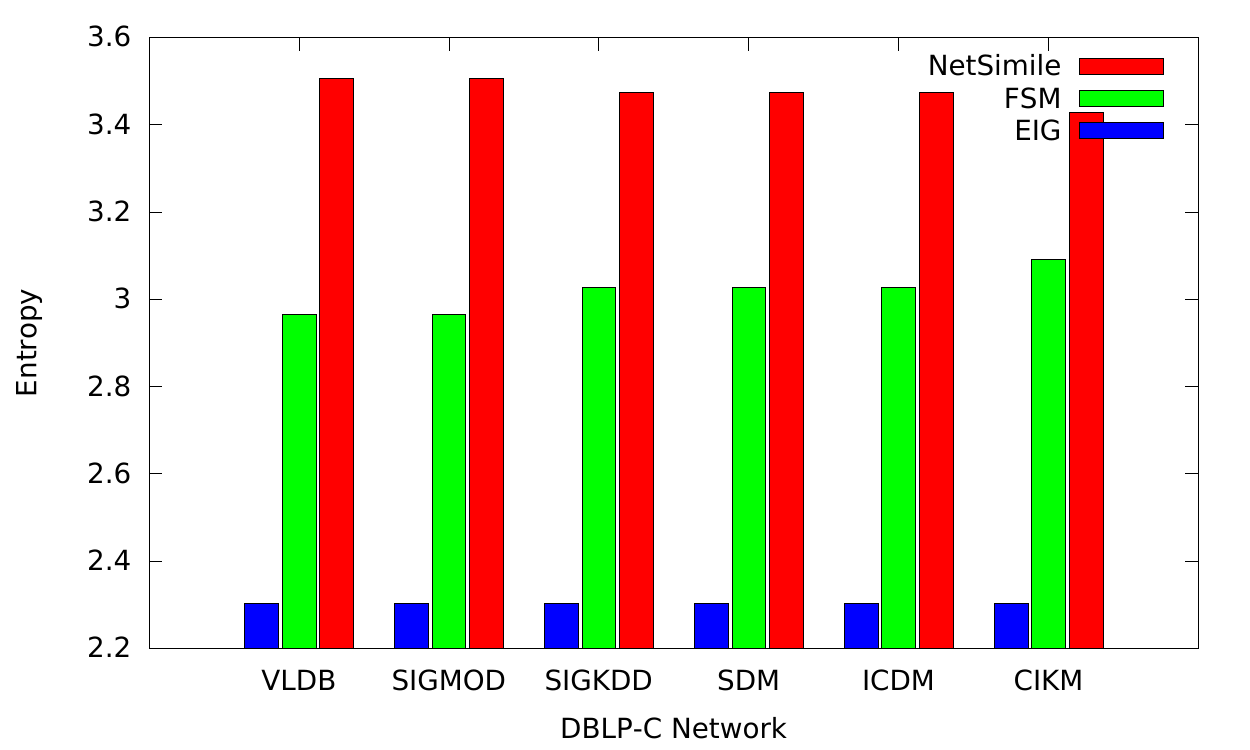}\\
\caption{Entropy of feature vectors generated by \simgraph, \fsm, and \eig on the DBLP-C co-authorship networks.  \simgraph's feature vectors have higher entropy than \fsm's or \eig's, which implies that they are capturing the nuances in the graphs better than \fsm or \eig.}
\label{fig:entropy}
\end{figure}

Figure~\ref{fig:heatmaps}(a) depicts results of \simgraph (Scaled Canberra Distance) on all datasets described in Section~\ref{subsec:data}.  Figure~\ref{fig:heatmaps}(b) shows the maximum p-values on \simgraph's local-feature distributions obtained by running the Mann-Withney U Test. The former (\simgraph Scaled Canberra Distance) is more discriminative in teasing out graph similarity than the latter (Mann-Whitney U Test). For instance, Figure~\ref{fig:heatmaps}(b) has many black grids corresponding to a maximum p-value of 0. So, even though some pairwise comparisons are possible, the Mann-Whithney U Test is not able to capture differences between \emph{any} two networks.  For brevity, we have omitted the average p-value results, which were comparable with the maximum p-value results. 

\begin{table}[tp]
\centering\small
\scalebox{1}{
\hspace{-2mm}
\begin{tabular}{|l|c|}
\hline
\multirow{3}{*}{{\bf Cluster 1}} & o-1 o-2 o-3 o-4 o-5 \\
& ql-1 ql-2 ql-3 ql-4 ql-5 \\
& r-1-bara r-10-bara r-100-bara \\
\hline \hline
\multirow{3}{*}{{\bf Cluster 2}} & r-1-er r-10-er r-100-er\\
&  r-1-ff r-10-ff r-100-ff \\
& r-1-ws r-10-ws r-100-ws \\
\hline \hline
{\bf Cluster} 3 & d-vldb d-sigmod d-cikm d-sigkdd d-icdm d-sdm \\
\hline \hline
\multirow{3}{*}{{\bf Cluster 4}} & a-AstroPh a-CondMat a-GrQc a-HepPh a-HepTh \\
& d-05 d-06 d-07 d-08 d-09\\
&  i-05 i-06 i-07 i-08 i-09  \\
\hline
\end{tabular}
}
\caption{\simgraph with \emph{x}-means classifies the networks in an intuitive way.}\label{tab:xmeans}\vspace{-3mm}
\end{table}

\subsection{Interpretability of Results}
To make sense of our results, we exploit the background knowledge about the networks used in our experiments.  Amid the real networks, we have three sets of collaboration networks (DBLP-C, DBLP-Y and IMDb), one technological network (Oregon AS), and a word co-occurrence network (Query Log). In addition, we have different synthetic networks generated by various commonly used models. One would expect these networks to be ``clustered'' by their types. This idea was inspired by the considerations found in \cite{Newman03thestructure}, where a large set of networks of different types are analyzed, together with their typical global and local features. For these experiments, we use two clustering algorithms: (1) agglomerative clustering~\cite{hierclustering} with Canberra Distance and unweighted average linking and (2) \emph{x}-means clustering~\cite{xmeans}.  We chose the former since hierarchical clustering allows for easy interpretation of results.  We chose the latter because it is a nonparametric version of \emph{k}-means, where the number of clusters $k$ is picked automatically through model selection.

Figure~\ref{fig:dendro}(a) presents the dendrogram of all of our networks built by hierarchical agglomerative clustering with unweighted average linking and the Canberra Distance and using \simgraph's graph ``signature'' vectors. The network names are colored by data set. As evident in Figure~\ref{fig:dendro}(a), there is a clear distinction between the clusters.  The collaboration networks appear all together, along with the forest fire synthetic networks.  The Oregon AS forms a cluster that only at the height of 0.45 joins with the Query Log.  The Erd\"{o}s-R\'enyi and Watts-Strogatz form a separate cluster. This, in turns, reflects our aforementioned intuition about following our background knowledge of the data. Similar results are obtained by applying the \emph{x}-means clustering on the vectors of local features, for which we report the outcome in Table \ref{tab:xmeans}. A part from the distribution of the random networks, the clusters reflect what we observe in Figure \ref{fig:dendro}(a). 

Figure \ref{fig:dendro}(b) shows the dendrogram for the above experiment (hierarchical agglomerative clustering with unweighted average linking and the Canberra Distance) for graph vectors generated by \eig.  This figure clearly shows a different picture, where the networks are grouped differently (see how the distribution of the colors is mixed). For example, in the leftmost cluster, two collaboration networks from arXiv are put together with four Query Log networks, while the missing Query Log network is placed together with the Oregon AS networks. The \eig results are not intuitive, thus making \eig not suitable for interpreting graph-similarity results.

\begin{figure*}[tp]
\centering
\begin{tabular}{cc}
\hspace{-7mm}\includegraphics[scale=0.45]{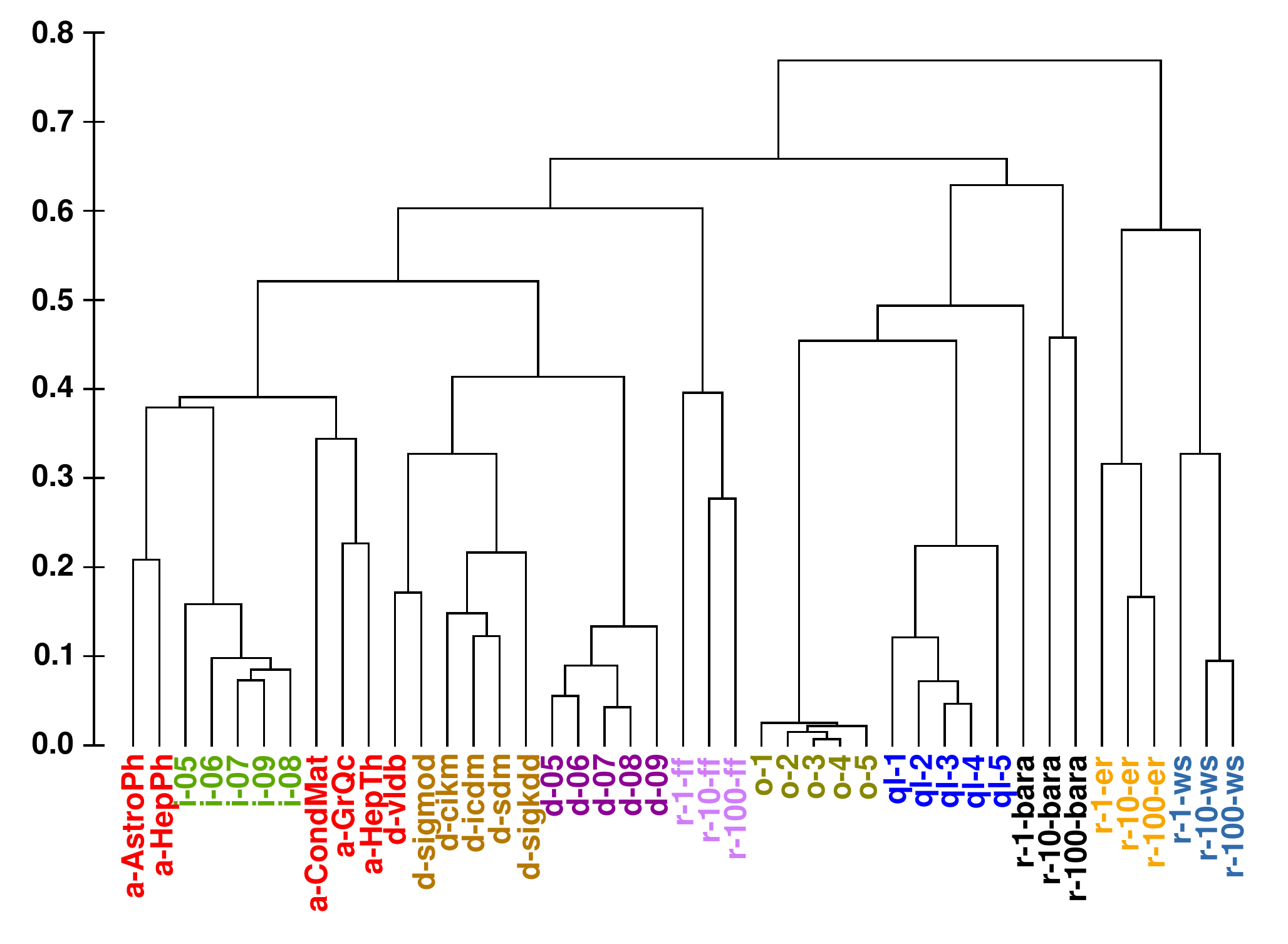}
&
\hspace{-3mm}\includegraphics[scale=0.45]{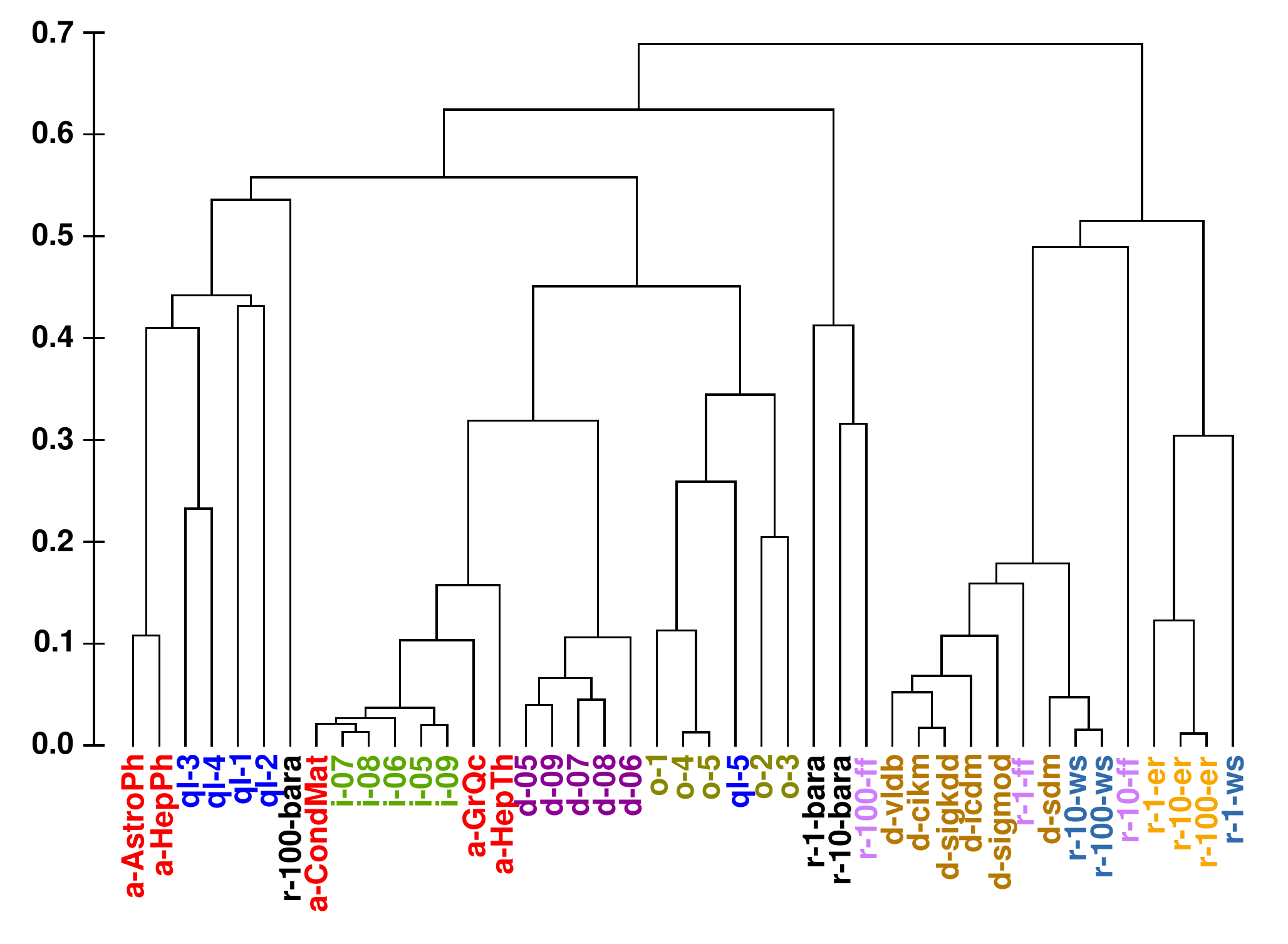}\vspace{-4mm}\\
(a) & (b)\vspace{-2mm}\\
\end{tabular}
\caption{Hierarchical dendrograms of all network based on (a) \simgraph with Canberra Distance, and (b) \eig with Canberra Distance.  Network names are colored by data set.  Homogeneity in colors (\simgraph's dendrogram) indicates better and more intuitive groupings (than \eig's dendrogram).}\label{fig:dendro}
\end{figure*}

\begin{figure}[b!]
  \centering\small
  \begin{tabular}{c}
\hspace{-7mm}\includegraphics[width=1.1\columnwidth]{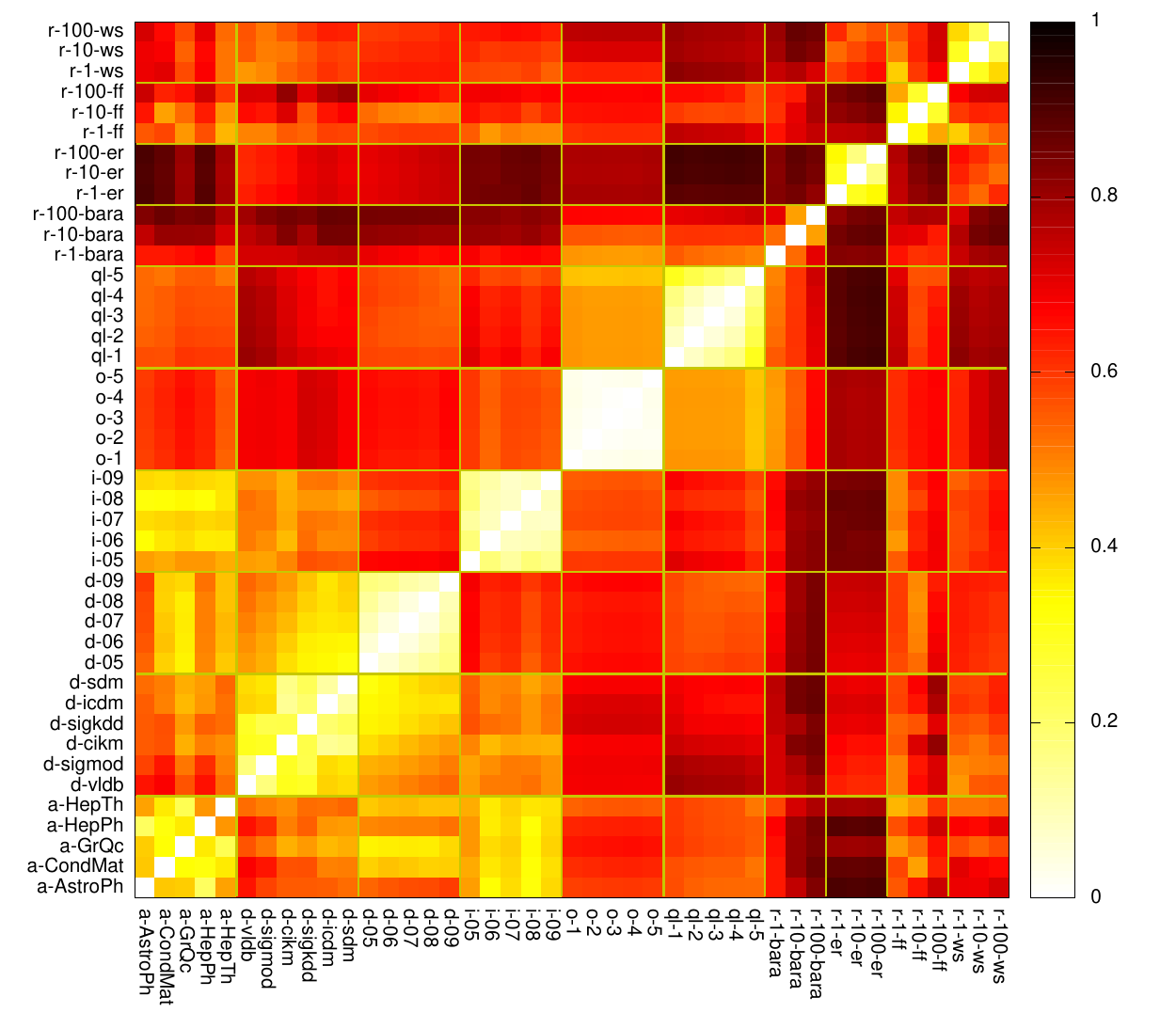}\\
(a) \simgraph with Canberra Distance\\
\hspace{-7mm} \includegraphics[width=1.1\columnwidth]{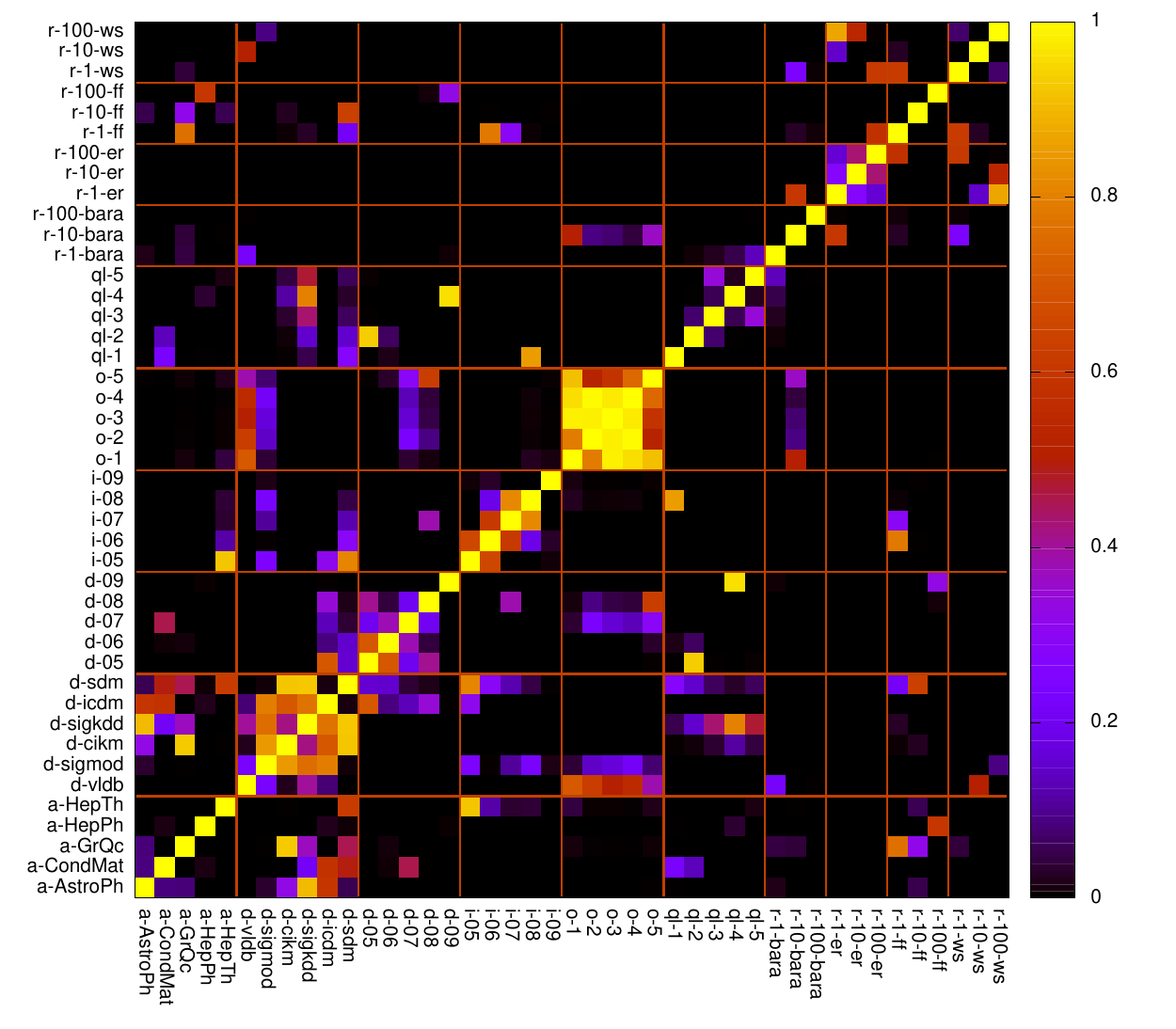}\\
(b)Mann-Whitney U Test:\\ Maximum p-values on Local-Feature Distributions
  \end{tabular}
  \caption{\simgraph outperforms the Mann-Whitney U test on local feature distributions by producing more discriminative results.  Plots are heatmaps of scores of all pairs of networks. Grid lines indicate ground truth, marking groups of networks. The ideal methods should have high scores (white in (a), yellow in (b)) on the diagonal blocks.}
  \label{fig:heatmaps}
  
\end{figure}

\subsection{Similarity of Networks with Different Sizes}
One question that may arise regarding \simgraph is whether its results are affected by the differences in sizes or other basic statistics of the two networks being compared. We do not want the size to play an important role in our solutions given that our interpretation of the question ``are two networks similar?'' leads to the question ``do the two networks follow the same (or similar) underlying linking model?''. 

To answer the aforementioned questions, we compared the relationships between the \simgraph with Canberra Distance and some basic statistics of our real and synthetic networks.  Specifically, we compared \simgraph values of two networks with the ratio between their (1) number of nodes, (2) number of edges, (3) average clustering coefficients of the nodes, (4) average degree, (5) maximum degree, and (6) network clustering coefficient.  In all of them, we saw no correlation.  For brevity, we only show the scatterplot for the \simgraph values and the ratio between the number of nodes of the two networks (see Figure~\ref{fig:coscansize}(a)) and the scatterplot for the \simgraph values and the ratio between the average clustering coefficients of the nodes of the two networks (see Figure~\ref{fig:coscansize}(b)).  As evident in these scatterplots, \simgraph's results are not merely reflecting the difference in sizes of the networks. If they were, we would expect to observe  correlations among the points in each scatterplot. This implies that we can generate two networks of the same kind, with different sizes (e.g., two Forest-Fire networks~\cite{generator-FF} of sizes 10K and 100K nodes) and \simgraph would find them similar.

\begin{figure}[t!]
\centering\small
\begin{tabular}{cc}
\hspace{-3mm}\includegraphics[width=.5\linewidth]{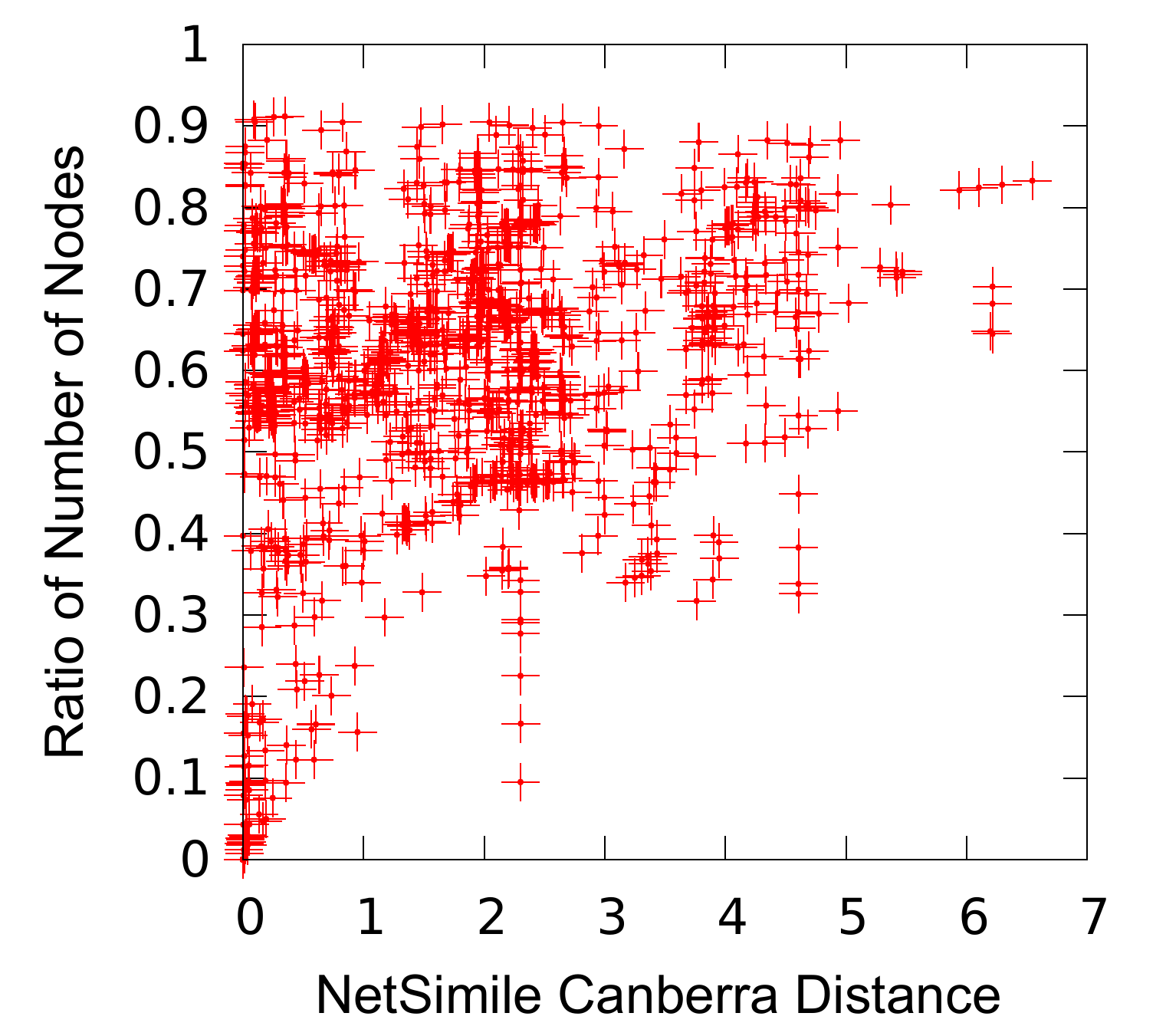} &
\hspace{-1mm}\includegraphics[width=.5\linewidth]{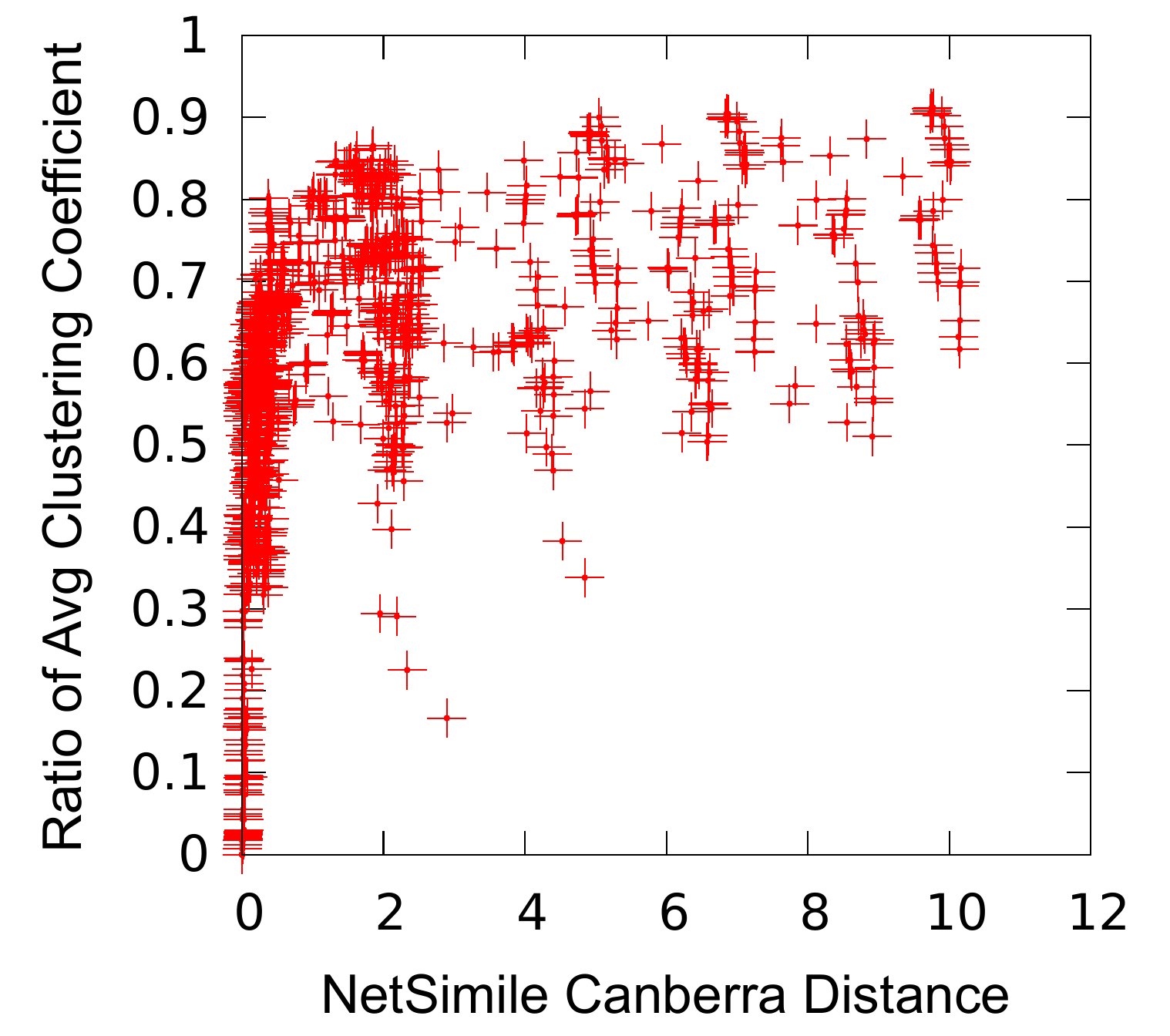}\\ 
(a) \simgraph vs. &
(b) \simgraph vs. ratio\\
ratio of \# nodes &
of avg.~clustering coefficient\\
\end{tabular}
\caption{\simgraph Canberra Distance is not measuring size, as there is no clear evidence of correlation between the two axes.}\label{fig:coscansize}
\end{figure}

\subsection{Scalability}

Table \ref{tab:scalability} reports the run times (in seconds) of \simgraph and the two baselines (\fsm and \eig) when applied to our real networks. Note that for \simgraph the run times refer to all the three steps described in Section~\ref{sec:meth}, with the comparison step constituted by the pairwise computation of both the Cosine Similarity and the Canberra Distance.  For \fsm we do not report running time longer than two days. 

\simgraph and \eig are able to compare graphs in a matter of seconds, though \eig produces results that are size-dependent. \fsm pays for its subgraph isomorphism, which considerably affects the performances. Note that \fsm is affected not only by the size of the network, but also by its type. While DBLP is a set of collaboration networks (with sparsely connected cliques), the Oregon AS (being a technology network) is made of one single connected component, thus the cost for the isomorphism becomes much higher. 

\begin{table}[t!]
      \centering\footnotesize\setlength{\tabcolsep}{1.5pt}
\scalebox{1}{
\begin{small}
      \begin{tabular}{|c|c||r|r|c|c|c|}
        \hline
        \multicolumn{2}{|c||}{\bf Network} &  \multicolumn{1}{c|}{\bf $|\text{V}|$} &  \multicolumn{1}{c|}{\bf $|\text{E}|$} &  \multicolumn{1}{c|}{\bf \simgraph}  & \multicolumn{1}{c|}{\bf FSM} & \multicolumn{1}{c|}{\bf EIG}\\
        \hline \hline
\multirow{5}{*}{\begin{sideways}arXiv\end{sideways}} 
& a-AstroPh &18,772 & 396,160 & 9 &$>$ 2 days & 6 \\ 
& a-CondMat & 23,133 & 186,936 & 2 & $>$ 2 days & 4 \\
& a-GrQc & 5,242 & 28,980 & 1 &$>$ 2 days & 1\\
& a-HepPh & 12,008 & 237,010 & 6  &$>$ 2 days&3 \\
& a-HepTh & 9,877 & 51,971 & 1 & $>$ 2 days&2 \\\hline
\multirow{6}{*}{\begin{sideways}DBLP-C\end{sideways}}
& d-vldb & 1,306 & 3,224 & 1 & 15 & 1\\
& d-sigmod & 1,545 & 4,191 & 1 & 28 &1\\
& d-cikm & 2,367 & 4,388 & 1 & 11 &1 \\
& d-sigkdd & 1,529 & 3,158 & 1 & 42 &1 \\
& d-icdm & 1,651 & 2,883 & 1 & 17 &1\\
& d-sdm & 915 & 1,501 & 1 & 7 &1\\\hline
\multirow{5}{*}{\begin{sideways}DBLP-Y\end{sideways}}
&  d-05 & 39,357 & 79,114 & 1 & 2231 &2\\
&  d-06 & 44,982 & 94,274 & 1 & 2856 &2 \\
&  d-07 & 47,465 & 103,957 & 1 & 4603 &2 \\
&  d-08 & 47,350 & 107,643 & 1 & 9859 &3 \\
&  d-09 & 45,173 & 102,072 &  1& 9209 &2\\\hline
\multirow{5}{*}{\begin{sideways}IMDb\end{sideways}}
&  i-05 & 13,805 & 130,295 & 1 & $>$ 2 days&3\\
&  i-06 & 14,228 & 142,955 & 1 & $>$ 2 days&3\\
&  i-07 & 13,989 & 133,930 & 1&$>$ 2 days &2 \\
&  i-08 & 14,055 & 132,007 & 1 &$>$ 2 days&3 \\
&  i-09 & 14,372 & 128,926 & 1 &$>$ 2 days& 2\\\hline
\multirow{5}{*}{\begin{sideways}Oregon AS\end{sideways}}
& o-1 & 10,900 & 31,181 & 2 &$>$ 2 days& 1\\
& o-2 & 11,019 & 31,762 &2 & $>$ 2 days& 1\\
& o-3 & 11,113 & 31,435 & 2 &$>$ 2 days& 1\\
& o-4 & 11,260 & 31,304 & 2 &$>$ 2 days& 1\\
& o-5 & 11,461 & 32,731 &  2&$>$ 2 days &1\\\hline
\multirow{5}{*}{\begin{sideways}Query Log\end{sideways}}
& ql-1 & 138,976 & 1,102,606& 209 & $>$ 2 days & 14\\
& ql-2 & 108,420 & 876,517 & 119 & $>$ 2 days & 11\\
& ql-3 & 89,406 & 707,579 & 107 &$>$ 2 days& 9\\
& ql-4 & 75,838 & 582,703 & 68 &$>$ 2 days& 8\\
& ql-5 & 42,946 & 253,469 & 11 & $>$ 2 days&5\\\hline
\end{tabular}
\end{small}
}
\caption{Run times (in seconds, unless otherwise noted) of \simgraph, \fsm, and \eig on our real networks}
\label{tab:scalability}
\end{table}

\subsection{Applications}

\simgraph can be used in numerous graph mining applications.  Here we discuss three of them.

\textbf{\simgraph as a Measure of Node-Overlap.}  Given three graphs $G_A$, $G_B$, and $G_C$ of the same domain (e.g., co-authorship networks in \emph{SIGMOD}, \emph{VLDB} and  \emph{ICDE}), can we use \emph{only} their \simgraph's ``signature'' vectors to gauge the amount of node-overlap between them?  Our hypothesis is that if graph $G_A$ is more similar to graph $G_B$ than graph $G_C$, then $G_A$ will have more overlap in terms of nodes with $G_B$ than $G_C$.  To test this hypothesis, we ran \simgraph with Canberra Distance on our real networks.  Figure~\ref{fig:nodeoverlap}(a) depicts the scatterplot of \simgraph results on graphs within each comparable group (i.e., arXiv, DBLP-C, DBLP-Y, IMDb, Query Log, and Oregon AS graphs).  The $y$-axis is the normalized node overlap and is equal to $\frac{\mid V_{G_A} \cap V_{G_B} \mid}{\sqrt{\mid V_{G_A} \mid \times \mid V_{G_B} \mid}}$. As the figure shows the lower the \simgraph Canberra Distance, the higher the normalized node intersection.  This confirms our hypothesis that \simgraph can be used to gauge node-overlap between two graphs without node correspondence information. Figure~\ref{fig:nodeoverlap}(b) shows the same scatter plot, but computed using the EIG Canberra Distance approach. In this case, there is no correlation between node overlap and the distance. Due to its scalability issues, the FSM approach could not be computed on all the networks in Figure~\ref{fig:nodeoverlap}.

\begin{figure}[!h]
\centering
\begin{tabular}{c}
\hspace{-4mm}\includegraphics[scale=1.25]{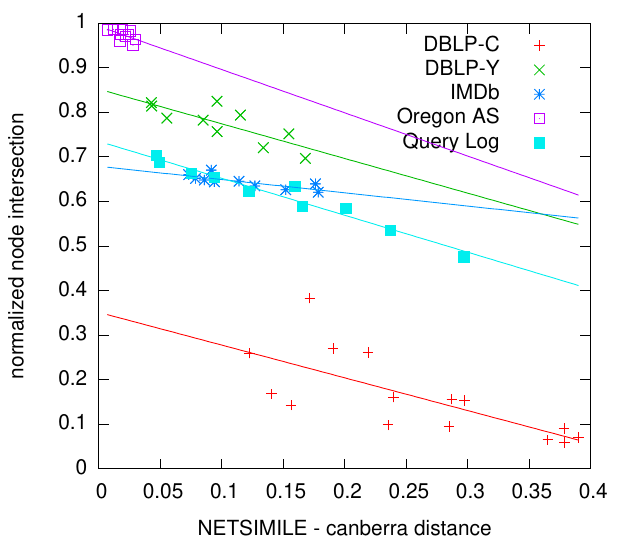}\\
(a)\\
\hspace{-4mm}\includegraphics[scale=1.25]{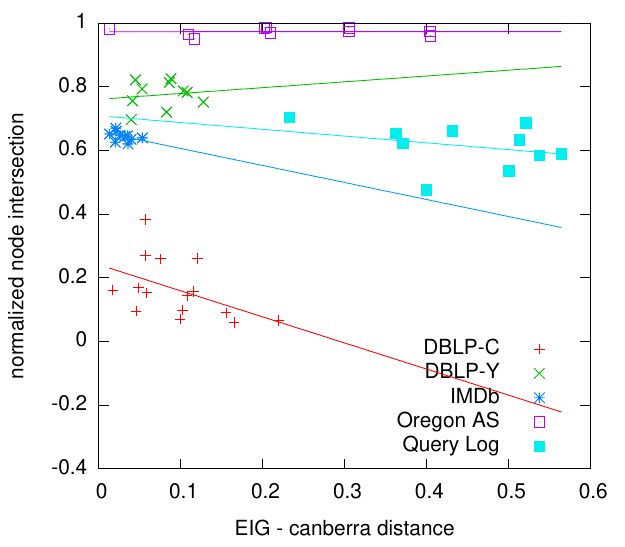}\\
(b)
\end{tabular}
\caption{(a) \simgraph Canberra Distance
  on DBLP, IMDb, Oregon and QueryLog. (b) EIG Canberra Distance on the
  same networks. \simgraph is an effective measure for node overlap without
  any node-correspondence information.  The lower the \simgraph Canberra
  Distance, the higher the normalized node intersection.  This correlation does not
  hold for EIG.  The points in both plots are along the fitted lines. For \simgraph (a), the 
 root mean square of residuals are $6.5E{-2}$ for DBLP-C,
  $2.6E{-2}$ for DBLP-Y,  $9.0E{-3}$ for IMDb,  $1.4E{-2}$ for Oregon AS, and
  $6.5E{-2}$ for Query Log. For EIG (b), the root mean square of residuals are $8.2E{-2}$ for DBLP-C,
  $4.2E{-2}$ for DBLP-Y,  $1.3E{-3}$ for IMDb, $1.2E{-2}$ for Oregon AS, and $6.7E{-2}$ for Query Log.}
\label{fig:nodeoverlap}
\end{figure}

\textbf{\simgraph as a Network Labeler.} Given a new (\emph{never before seen}) graph, can we use the Canberra Distance between its \simgraph's ``signature'' vector to known graphs' \simgraph ``signature'' vectors to accurately predict its label?  To answer this question, we setup and ran the following 4-step experiment.  In step 1, we created a set of \emph{test} graphs by generating 50 synthetic graphs of types Erd\"{o}s-R\'{e}nyi, Watts-Strogatz, Barab\'{a}si, and Forest Fire.  In step 2, for each test graph, we compared its \simgraph score using the normalized Canberra Distance with existing graphs (as reported in Table~\ref{tab:networks}).  In step 3, we assigned to the test graph the label of its most similar graph.  In step 4, we computed the accuracy of our predictions.  

The predictive accuracy of \simgraph was 100\% -- i.e., \simgraph was able to label all 50 test graphs accurately.  For each of the 50 test graphs, we inspected the \simgraph normalized Canberra Distance between the most similar graph (whose label we chose) and the second most similar graph (whose label we did not choose).  Let's call the former $dist_1$ and the latter $dist_2$.  The minimum difference between $dist_1$ and $dist_2$ across the 50 test graphs was 0.001.  The maximum was 0.428.  The mean difference was 0.143; and the standard deviation was 0.112.  Thus, the answer to the aforementioned question of whether \simgraph can be used effectively as a network labeler is yes.

We ran the same experiments using EIG with Canberra Distance on the same networks. The predictive accuracy of EIG was 72\%, i.e., 14 graphs were uncorrectly labeled. Figure~\ref{fig:ranklabel} shows the distribution of the ranking for the correct labels of graphs.  There are two cases, in which the correct (i.e. true) labels for the graphs are ranked 11th by EIG. Due to scalability issues, FSM could not be performed on all the networks in this experiment.

\begin{figure}[!h]
\centering
\includegraphics[scale=1.1]{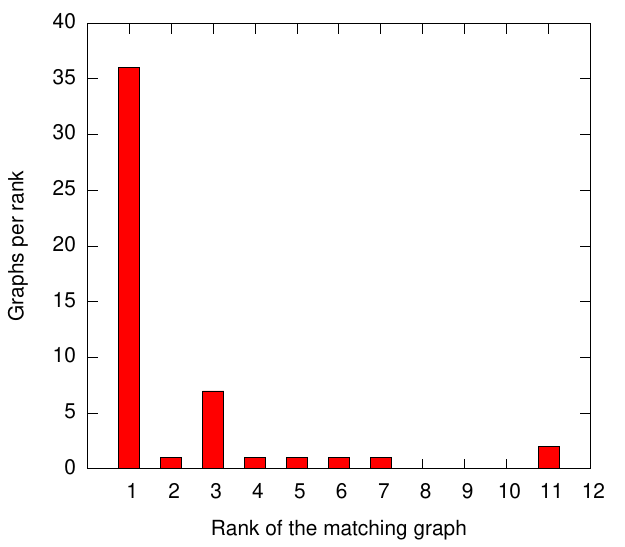}\\
\caption{Distribution of the rankings of the correct labels for the synthetic graphs by EIG with Canberra Distance.  EIG's accuracy is only 72\%.  There are two graphs whose correct  labels are ranked 11.}
\label{fig:ranklabel}
\end{figure}

\textbf{\simgraph as a Discontinuity Detector.} Given a time-series of graphs $\{G_1, G_2, G_3, ..., G_t\}$, can \simgraph detect any discontinuity (i.e. temporal outliers) present in the data?  To answer this question, we utilize \simgraph Canberra Distance to compute the difference between graphs in a time series.  For this experiment, we used data coming from two different messaging services, Yahoo! IM and Twitter.

The first dataset contains 28 days of Yahoo! IM communications (\url{http://sandbox.yahoo.com}), starting Tuesday, April 1, 2008. Each graph is a collection of instant messages (IMs) per day, with nodes representing IM users and links denoting communication events.  The graphs are of varying sizes: number of nodes from 29K to 100K and number of edges from 80K to 280K.  We computed the \simgraph normalized Canberra Distance between Day 0 (April 1, 2008) and the other 27 days. Figure~\ref{fig:discontinuity}(a) shows our results, with the $x$-axis representing days and the $y$-axis representing \simgraph (with normalized Canberra Distance) between Day 0 and the other 27 days.  Figure~\ref{fig:discontinuity}(b) shows \simgraph (with normalized Canberra Distance) between Day 8 (April 9, 2008) and all the other days.  As the figures illustrate, \simgraph detects the weekday vs. weekend discontinuities.  It also detects a discontinuity on Wednesday April 9, 2008. The following event explains this discontinuity. Flickr announced that it will add video to its popular photo-sharing community\footnote{\url{http://yhoo.client.shareholder.com/releasedetail.cfm?releaseid=303857}} on April 8, 2008; but its news spread on April 9, 2008.\footnote{\url{http://searchengineland.com/}\\ \url{flickr-launches-video-its-not-a-youtube-clone-13727}} This event is reflected in the graph for April 9, 2008, where the number of connected components decreases by $4\times$ as the news about Flickr spreads among the IM users. 

\begin{figure}[!h]
\centering
\begin{tabular}{c}
\includegraphics[scale=0.75]{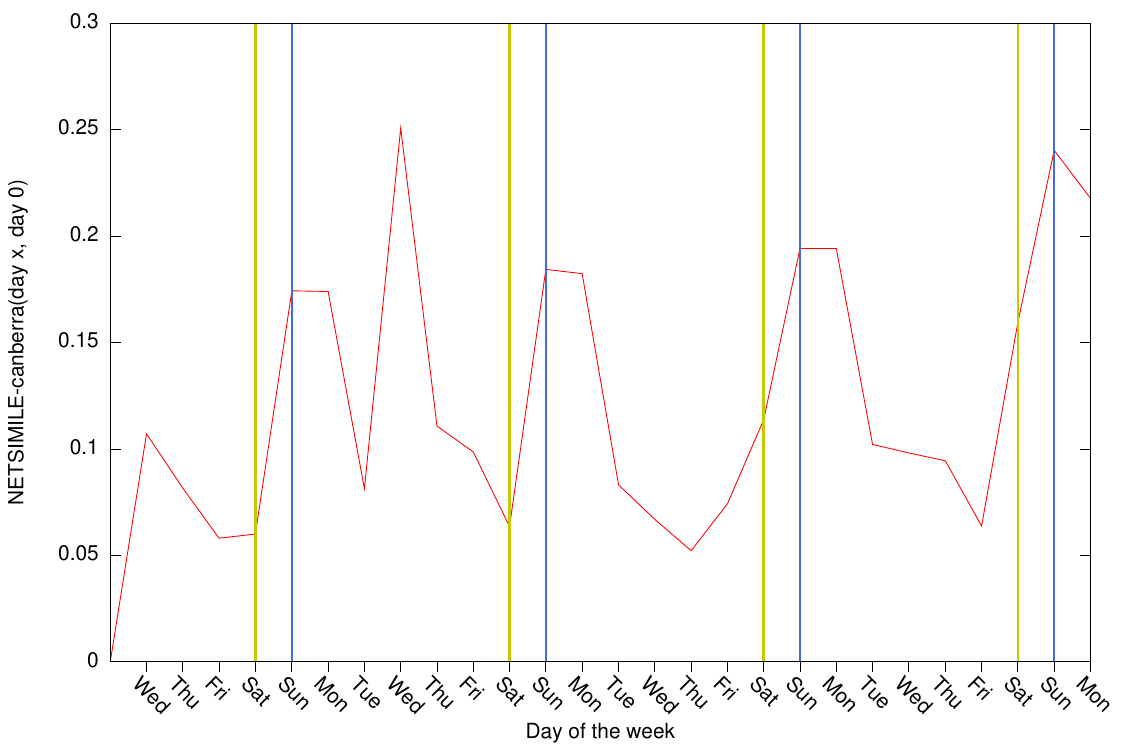}\\
(a) NetSimile between each day and day 0 in Yahoo! IM\\\ \\
\includegraphics[scale=0.75]{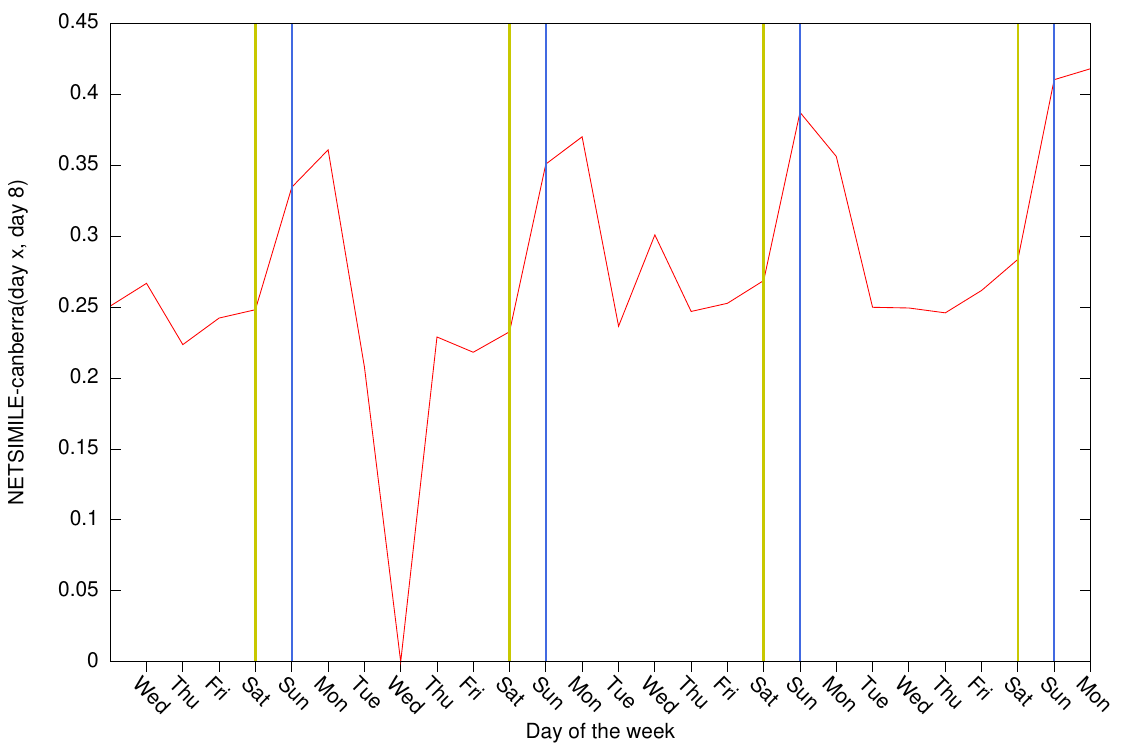}\\
(b) NetSimile between each day and day 8 in Yahoo! IM
\end{tabular}
\caption{\simgraph detects discontinuities in time-evolving graphs.  
(a) Distance of day 1, day 2, $\cdots$, day 27 IM graphs from day 0 (Tuesday April 1, 2008) IM graph. Weekdays are distinguished from weekends (yellow line=
Saturday, blue line = Sunday).  The peak on the 2nd Wednesday (April, 9, 2008) corresponds to a big Flicker announcement and a Microsoft offer to buy Yahoo!. (b) Distance of the other days from day 8 (April 9, 2008) IM graph. All the other days are distant from April 9, 2008.}
\label{fig:discontinuity}
\end{figure}

The second dataset contains 30 days of Twitter\footnote{\url{http://www.twitter.com}} @replies (i.e., messages that begin with a direct mention to ``@user''), starting Monday, June 1, 2009. The sizes of these graphs vary less than in the Yahoo! data: number of nodes range from 4K to 8K, and number of edges range from 2K to 5K. Similar to the Yahoo! experiments (detailed above), we computed the \simgraph normalized Canberra Distance between Day 0 and the other 30 days.  Figure~\ref{fig:discontinuity2}(a) shows our results, with the $x$-axis representing days and the $y$-axis representing \simgraph (with normalized Canberra Distance) between Day 0 and the other 30 days.  Figure~\ref{fig:discontinuity2}(b) shows \simgraph (with normalized Canberra Distance) between Day 6 (June 7, 2009) and all the other days. Figure~\ref{fig:discontinuity2} shows no particular periodicity.  This is not surprising since  Twitter @replies are semi-private conversations between two people and their common followers  (unlike instant messages that are private).  The @replies tweets have less of the ``news amplification'' effect as regular tweets \cite{sousa}.  However, NetSimile is still able to spot a significant discontinuity on day 6 (June 7, 2009). On that day, the largest cross-country election in the history took place (namely, the European parliamentary elections), affecting almost 500M people; parliamentary elections also took place in Lebanon; high school students graduated in the U.S.; and  Roger Federer became the sixth man in tennis history to complete a career Grand Slam  (by winning the French Open on that day) and tied Pete Sampras' Grand Slam record. 

\begin{figure}[!h]
\centering
\begin{tabular}{c}
\includegraphics[scale=0.75]{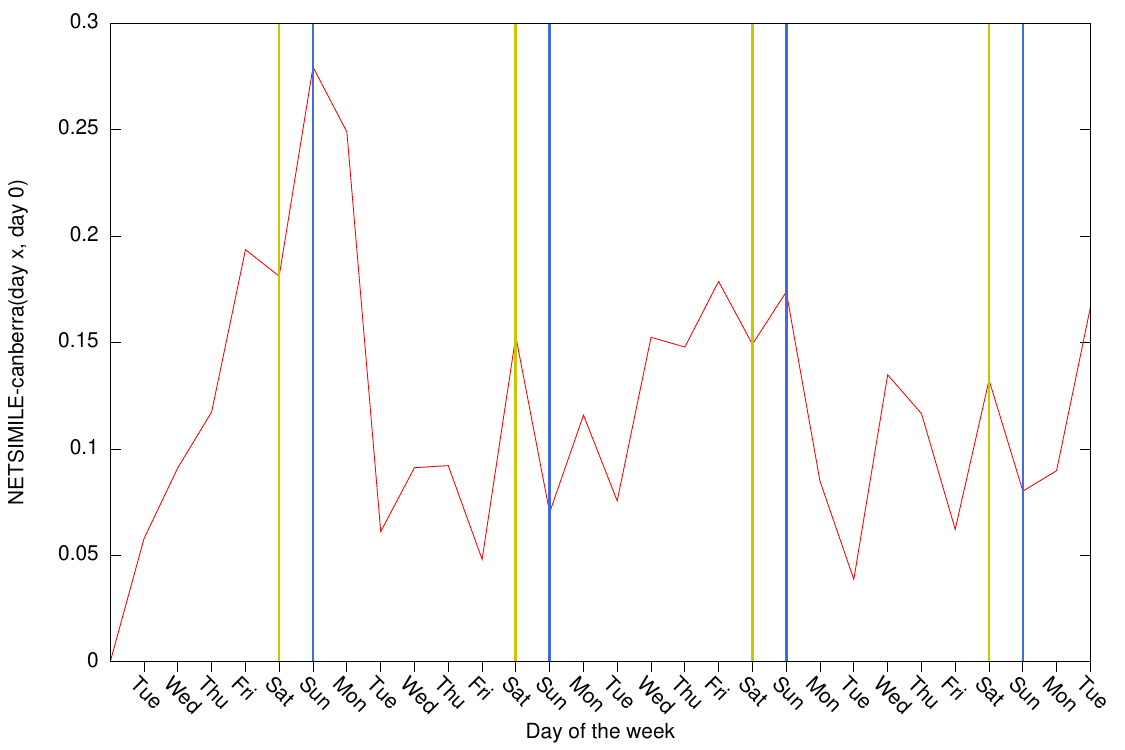}\\
(a) NetSimile between each day \& day 0 in Twitter @replies\\\ \\
\includegraphics[scale=0.75]{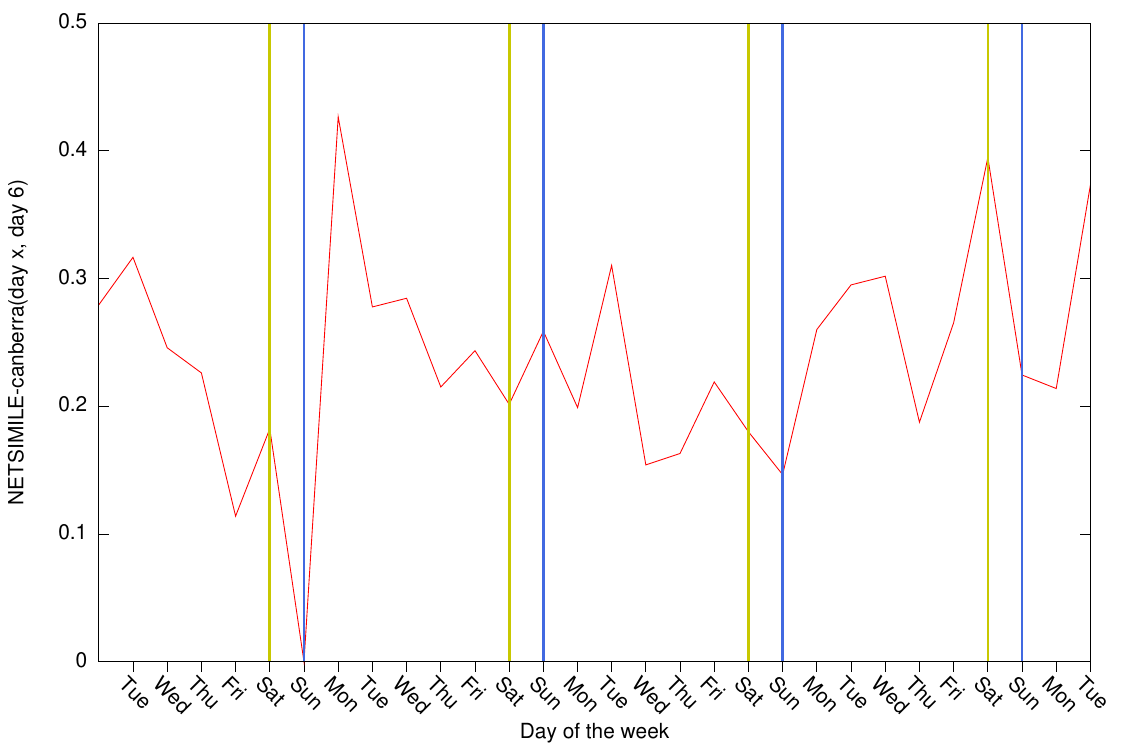}\\
(b) NetSimile between each day \& day 6 in Twitter @replies\\\ \\
\end{tabular}
\caption{\simgraph detects discontinuities in time-evolving graphs.   (a) Distance of day 1, day 2, $\cdots$, day 30 Twitter graphs from day 0 (Monday June 1 2009) graph. Weekdays are distinguished from weekends (yellow line= Saturday, blue line = Sunday). (b) Distance of the other days from day 6 (June 7, 2009) Twitter graph. All the other days are distant from June 7, 2009. Twitter @replies are semi-private conversations; thus, we do not expect to see periodicity  (like in IM conversations). On June 7, 2009, elections were held for the European Parliament and for the Lebanese Parliament, Roger Federer wins the French Open and makes tennis history, and high-schoolers graduate in the U.S.}
\label{fig:discontinuity2}
\end{figure}

\section{Related Work}
\label{sec:background}

Assessing the similarity between two ``objects'' comes up in numerous settings. Thus, the literature is rich in similarity measures for various domains: distributions or multi-dimensional points \cite{survey_sim_measures}, datacubes \cite{BRV11}, and graphs, such as social (\cite{faust, macindoe}), information \cite{web_graph_sim}, and biological networks \cite{codense}. Here, we focus on graph similarity when the node correspondence is unknown. The proposed methods can be divided into three main classes.

\textbf{(1) Graph isomorphism.} The similarity of the graphs depends on whether the graphs are isomorphic to each other  \cite{graph_isomorphism}; or one graph is a subgraph of the other (\cite{subgraph1, subgraph2});  or they have common subgraphs \cite{max_common_subgraph}. The exact algorithms of these problems are exponential, rendering them inapplicable to the large graphs on which our research focuses. Graph edit distance is a generalization of the graph isomorphism problem which consists of finding the minimum number of operations (insertions, deletions, renaming of nodes, reversion of edges) that is required to convert one graph to another. Different cost functions (\cite{edit_dist1, edit_dist2}) have been proposed in the literature. It is noteworthy to mention that graph isomorphism addresses mainly the graph matching problem (roughly, given two graphs, find the correspondence between their nodes).  It does not directly address the graph similarity problem, which is the topic of this paper.

\textbf{(2) Iterative methods.} Two well known representatives of this category are  SimRank \cite{SimRank02} and similarity flooding \cite{Melnik02}.  The former computes all the pairwise similarities between the nodes. The latter attempts to find the correspondences between the nodes of the graphs in order to assess their similarity. Zager et al.~\cite{ZV08} proposed a method that computes the similarity of the graphs by coupling the pairwise similarities between the nodes with the similarities between the corresponding edges.

\textbf{(3) Feature Extraction.}  These methods, which are based on comparing specific graph features, are popular due to their scalability.  The challenge here is the appropriate selection of features/patterns, since some features, such as frequent subgraphs (\cite{Berlingerio09, kumar_mining, codense, Karypis, gspan}), can be computationally expensive to extract.  Macindoe et al.~\cite{macindoe} and Faust \cite{faust} focus on comparing social networks by extracting socially relevant features. GraphGrep \cite{graphGrep} extracts paths by doing random walks on the graphs. GString \cite{gstring}  converts the graphs into sequences and extracts features from the latter.  G-Hash \cite{ghash} applies a wavelet matching kernel to indirectly extract information about the neighborhoods of the nodes. Papadimitriou et al.~\cite{web_graph_sim}  compare web graphs, where the correspondence of the nodes is known.  Henderson et al.~\cite{Refex11} propose a method for mining recursive structural features.  \simgraph is easily extensible to incorporate these features. Lastly, Li et al.~\cite{GengMBZ11} propose a classification approach of attributed graphs, which is based on global feature extraction. The weakness of the proposed method is that some features (e.g., eccentricity and shortest paths) are computationally expensive, and, thus, it is not scalable on large graphs. Moreover, the method is domain-specific and focuses in databases of graphs, such as chemical compounds, while our work aims at comparing graphs of different domains.

In this work, we focus on the \emph{local} feature extraction approach.  We aggregate a number of carefully chosen, interpretable, intuitive, and computationally inexpensive local features that capture the nuances in the structural information, and then employ various techniques (similarity measures \cite{survey_sim_measures}, hierarchical clustering, and hypothesis testing) in order to find the pairwise similarity scores of the given (possibly, cross-domain) networks.  

\section{Conclusions}
\label{sec:concl}
We introduced \simgraph, a novel, effective, size-independent, and scalable method for comparing large networks. \simgraph has three components: (1) feature extraction, (2) feature aggregation, and (3) comparison.  The heart of  our contribution is in components (1) and (2), where we discovered that  moments of distributions of structural features computed on the nodes and their egonets provide an excellent ``signature'' vector  for a graph.  These ``signature'' vectors can be used to effectively and quickly assess the similarity of two or more graphs. 

Our broader contributions are:
\begin{itemize*}
\item \emph{Novelty}: \simgraph avoids the (expensive) node correspondence problem, as well as adjusts for graph size.
\item \emph{Effectiveness}: \simgraph gives results that agree with intuition and the ground-truth.
\item \emph{Scalability}: \simgraph generates its ``signature'' vectors in time linear on the input size (i.e., number of edges of the input graphs).
\item \emph{Applicability}: \simgraph's ``signature'' vectors are useful in numerous graph mining tasks.  In addition, \simgraph is easily extensible to include features and aggregators besides the ones presented. 
\end{itemize*}

\bibliographystyle{IEEEtran}

\end{document}